\documentclass[11pt]{article}
\usepackage{graphicx} % Required for inserting images
\usepackage{amssymb}
\usepackage{amsthm}
\usepackage{amsmath}
\usepackage{xcolor}
\usepackage{subcaption}
\usepackage{multicol}
\usepackage{multirow}
\usepackage{array}
\usepackage{tabularx}
\usepackage{booktabs}
\usepackage{hyperref}
\usepackage[letterpaper, margin = 1in]{geometry}
\usepackage{titling,lipsum}
\usepackage{thmtools} 
\usepackage{thm-restate}
\usepackage{authblk}
\usepackage{hyperref}

\title{New Algorithms and Hardness Results for Connected Clustering\thanks{This work has been funded by the Deutsche Forschungsgemeinschaft (DFG, German Research Foundation) – 390685813; 459420781 and by the Lamarr Institute for Machine Learning and Artificial Intelligence lamarr-institute.org.}}

\author[1]{Jan Eube}
\author[2]{Heiko R\"oglin}
\affil[1]{University of Bonn, Germany, \href{mailto:eube@cs.uni-bonn.de}{\texttt{eube@cs.uni-bonn.de}}}
\affil[2]{University of Bonn, Germany, \href{mailto:roeglin@cs.uni-bonn.de}{\texttt{roeglin@cs.uni-bonn.de}}}

\date{\vspace{-5ex}}
\usepackage[ruled,linesnumbered]{algorithm2e}

\newcommand{\U}{\mathcal{U}}

\newcommand{\Psc}{\mathcal{P}}
\newcommand{\cla}{Y}
\newcommand{\Com}{C}
\newcommand{\B}{\mathcal{B}}
\newcommand{\A}{\mathcal{A}}
\newcommand{\opt}{\textsc{Opt}}
\newcommand{\mycomment}[1]{}

\newtheorem{theorem}{Theorem}
\newtheorem{invariant}[theorem]{Invariant}
\newtheorem{lemma}[theorem]{Lemma}
\newtheorem{corollary}[theorem]{Corollary}
\newtheorem{observation}[theorem]{Observation}
\newtheorem{definition}[theorem]{Definition}

\DeclareMathOperator*{\argmin}{argmin}
\begin{document}

\begin{titlingpage}

\bibliographystyle{plainurl}
%\title{New Algorithms and Hardness Results for Connected Clustering\footnote{This work has been funded by the Deutsche Forschungsgemeinschaft (DFG, German Research Foundation) – 390685813; 459420781 and by the Lamarr Institute for Machine Learning and Artificial Intelligence lamarr-institute.org.}}
\maketitle
%\thanks{This work has been funded by the Deutsche Forschungsgemeinschaft (DFG, German Research Foundation) – 390685813; 459420781 and by the Lamarr Institute for Machine Learning and Artificial Intelligence lamarr-institute.org.}
\begin{abstract}

Connected clustering denotes a family of constrained clustering problems in which we are given a distance metric and an undirected connectivity graph $G$ that can be completely unrelated to the metric. The aim is to partition the $n$ vertices into a given number $k$ of clusters such that every cluster forms a connected subgraph of $G$ and a given clustering objective gets minimized. The constraint that the clusters are connected has applications in many different fields, like for example community detection and geodesy.

So far, $k$-center and $k$-median have been studied in this setting. It has been shown that connected $k$-median is $\Omega(n^{1- \epsilon})$-hard to approximate which also carries over to the connected $k$-means problem, while for connected $k$-center it remained an open question whether one can find a constant approximation in polynomial time. We answer this question by providing an $\Omega(\log^*(k))$-hardness result for the problem. Given these hardness results, we study the problems on graphs with bounded treewidth. We provide exact algorithms that run in polynomial time if the treewidth $w$ is a constant. Furthermore, we obtain constant approximation algorithms that run in FPT time with respect to the parameter $\max(w,k)$.

Additionally, we consider the min-sum-radii (MSR) and min-sum-diameter (MSD) objective. We prove that on general graphs connected MSR can be approximated with an approximation factor of $(3 + \epsilon)$ and connected MSD with an approximation factor of $(4 + \epsilon)$. The latter also directly improves the best known approximation guarantee for unconstrained MSD from $(6 + \epsilon)$ to $(4 + \epsilon)$.

\end{abstract}

\end{titlingpage}

\section{Introduction}

Clustering describes a family of problems that given a set of data points aim at grouping together similar points into a limited number of groups called clusters. These kind of problems have been studied widely both in practice as well as in theory.  There are numerous objective functions to evaluate the quality of a clustering, among them $k$-center, $k$-median, $k$-means, and min-sum-radii (MSR). In all of the listed objectives, the data points are usually contained in a metric space  which defines a distance function and the number of clusters $k$ is given beforehand. In addition to the clusters, one also needs to determine a center for each cluster that represents this cluster and the quality of the clustering is determined by the distance of the data points to their respective centers.

In $k$-center one aims to minimize the maximum distance among any data point to its respective center. In $k$-median the sum of the distances of all points to the centers and in $k$-means the sum of the squared distances to the centers should be minimized. The objective of min-sum-radii is determined by calculating for every cluster its so called radius, which is the maximum distance of any point in this cluster to the center, and the objective value is the sum of the radii of all clusters. There is also a variant of this problem called min-sum-diameter (MSD) where no centers are needed and instead of the radius one determines the diameter of each cluster, which is the maximum distance of two points in this cluster. In $k$-center, $k$-median, and min-sum-radii, the centers are usually chosen among the data points themselves while in the $k$-means problem we can often choose the centers arbitrarily if the data points are placed in the Euclidean space.  

All of these problems are NP-hard but there exist numerous approximation results. There exist multiple 2-approximation algorithms for $k$-center \cite{hochbaum1985best, Gonzalez85}, which is the best possible, unless P = NP \cite{hsu1979easy}. For $k$-median the currently best known algorithm achieves a $(2 + \epsilon)$-approximation guarantee \cite{cohen2approx} and for $k$-means there exists a $(9 + \epsilon)$-approximation algorithm for general metrics and a $(6.357 + \epsilon)$-approximation algorithm for Euclidean metrics~\cite{ahmadian2019better}. Recently, Buchem et al.\ provided a $(3 + \epsilon)$-approximation algorithm for MSR, which also directly implies a $(6 + \epsilon)$-approximation guarantee for MSD.

In many applications the clusters should fulfill additional requirements besides having a good objective value, which is why researchers have studied these objectives under various side constraints. For example, one can require that the clusters are not too small or too big which results in private \cite{aggarwal2010achieving} and capacitated clustering \cite{li2016approximating}. Additionally, different variants of fair clustering have been studied over the last decade \cite{chierichetti2017fair, vakilian2022improved, schmidt2020fair}. 

In this work we deal with a connectivity side constraint. In this setting, we are not only given a distance metric but the data points are also vertices in a so called connectivity graph which is independent of the metric. The aim is to find clusters minimizing the respective objective while ensuring that every cluster induces a connected subgraph of the connectivity graph. This side constraint has first been introduced by Ge et al.~\cite{ge2008joint}, who motivated it in the context of community detection. In this setting we are given two kinds of data for a group of entities (for example humans, organizations or events) called attribute and relationship data. The attribute data corresponds to the intrinsic properties of the entities and can be used to define the similarity of two entities. Ge et al.\ modeled this using the distance metric. The relationship data corresponds to extrinsic relationships between entities, for example a common paper of two researchers if we want to identify research communities. When identifying a community it seems reasonable to require that there are no parts of this community that are totally isolated from the rest. As a result, Ge et al.\ modeled this using the connectivity graph by adding edges between entities if they are related. Given that in many real world scenarios also very dissimilar entities may interact (for example joint work of researchers from different fields) while similarity between entities does not necessarily imply that there exists a relationship between them, it seems reasonable to not assume any clear correlation between the distance metric and the connectivity graph. Later also Gupta et al.~\cite{gupta2011clustering} and Eube et al.~\cite{eube2025connectedkmediandisjointnondisjoint} studied connected $k$-median motivated by this interpretation.

Another natural application of connected clustering is the clustering of geographic regions or objects. We are given information (e.g.\ regarding the population or climate) for a set of regions and want to cluster similar regions together while ensuring that the resulting clusters actually correspond to continuous areas. Validi et al.\ considered connected clustering in the context of finding political districts \cite{validi2022imposing} and Liao and Peng to cluster sensor data \cite{liao2012clustering}. Most recently, Drexler et al.\ used this model in the context of clustering tide gauge stations \cite{drexler2024connected}. Here one is given a sequence of measurements of the sea levels of different stations and wants to group stations with similar measurements together such that the groups are again reasonably compact.

Unfortunately it turned out that most clustering objectives get very hard to approximate once a connectivity constraint is introduced. Ge et al.\ already showed that even for $k= 2$ connected $k$-center cannot be solved optimally anymore \cite{ge2008joint}. While they claimed to have a constant approximation algorithm for connected $k$-center \cite{ge2008joint}, Drexler et al.\ later pointed out that this is not the case \cite{drexler2024connected}. The latter provided an $O(\log(k)^2)$-approximation algorithm, which is the best known guarantee to this day. They also provided constant approximations if the metric is an $L_p$ metric of constant dimension or has a constant doubling dimension\footnote{The approximation factor depends on the dimension in both cases.} but it remained an open question whether or not there exists a constant approximation in the general case. As we show in this work this is not the case,  unless P = NP, as connected $k$-center is $\Omega(\log^*(k))$-hard to approximate\footnote{$\log^*(k)$ denotes how often we have to iteratively apply the logarithmic function to $k$ until we end up with a number that is at most $1$.}.

The connected $k$-median problem turned out to be even harder, as Eube et al.~\cite{eube2025connectedkmediandisjointnondisjoint} showed that even for $k= 2$ connected $k$-median is $\Omega(n^{1-\epsilon})$-hard to approximate. The same reduction also works for connected $k$-means. They relaxed the problem by allowing vertices to be part of multiple clusters. In this setting they were able to provide an $O(k^2\log(n))$-approximation but also showed by slightly adapting a previous reduction by Gupta et al \cite{gupta2011clustering} that the problem remains $\Omega(\log(n))$-hard to approximate. Already before this, it was observed that one can slightly adapt the classic $k$-center algorithm by Hochbaum and Shmoys to obtain a $2$-approximation for connected $k$-center if the clusters may overlap \cite{ge2008joint}. However while this might be a reasonable relaxation for some settings, the possibility to assign an entity to multiple groups would appear unnatural in many situations. Especially if we cluster geographic data, having overlapping areas may not be reasonable in many contexts. Previous research also seems to focus primarily on the disjoint variant. In the remainder of this work we will focus exclusively on the disjoint setting.

While both connected $k$-center and $k$-median are very hard to approximate on general graphs, there exists dynamic programs for both problems that calculate the optimal solution in polynomial time if the graph is a tree \cite{ge2008joint,eube2025connectedkmediandisjointnondisjoint}. In this work, we extend these results for a broader class of `treelike' graphs. This is based on the concept of treewidth and tree decompositions defined by Robertson and Seymour \cite{robertson1984graph}, whose definition we present later. For many NP-hard problems it has been shown that they can be solved in polynomial time if the treewidth $w$ of the input graph is constant (see for example Chapter~$7$ of \cite{cygan2015parameterized}). By combining the dynamic programs for connected $k$-center and $k$-median with the ideas of an algorithm for Steiner trees on graphs with constant treewidth \cite{chimani2012improved}, we can also solve these problems exactly on graphs with constant treewidth. One may observe that especially in the setting that we want to cluster tide gauge stations as presented by Drexler et al., it is not unreasonable to assume that the graph is treelike. This is due to the fact that these stations are usually placed at the coast which means that for large parts there are sequences of stations just following the coastline. And while islands close to the mainland and almost enclosed bodies of water like the Mediterranean Sea can lead to cycles, the number of parallel cycles will not be too large. Also in other settings we might end up with a reasonably small treewidth.

However, the running times of the resulting algorithms lie in $\Omega(n^w)$. Using ideas similar to the nesting property for hierarchical clustering \cite{lin2010general}, we are able to obtain constant approximation algorithms, with a fixed parameter tractable (FPT) running time with respect to the parameter $\max(w,k)$ (i.e.\ its running time can be bounded by the product of a polynomial function in $n$ and another function that only depends on $\max(w,k)$).

Lastly we introduce the connectivity constraint to the MSR and MSD problem, which to the authors' knowledge has not been done before\footnote{There exists a paper by An and Kao \cite{an2024connected} that deals with the `Connected Minimum Sum of Radii Problem'. However their problem is not related to the connectivity constraint.}. 
%While there exists a previous paper by An and Kao \cite{an2024connected} that deals with the 'Connected Minimum Sum of Radii Problem', their problem is hardly related to the connectivity constraint as discussed in this work. They are only given a metric without a connectivity graph and calculate an arbitrary clustering as well as an Steiner Tree connecting the chosen centers to each other. The objective value is then a weighted sum of the min sum radii cost function and the length of the Steiner tree. It might be a little bit confusing that there are two different versions of connected clustering, but both notions are consistent with previous literature (for example \cite{swamy2004primal} used similar definition as An and Kao). In this work connected MSR always refers to the MSR problem under the connectivity side constraint which to the authors knowledge has not been studied before. 
While it is NP hard to obtain constant approximations for the other clustering objectives under the connectivity constraint, we show that the approximation algorithm by Buchem et al.~\cite{buchem20243+} can be modified to also ensure connectivity which results in a polynomial time $(3 + \epsilon)$-approximation algorithm for connected MSR. Additionally, we noticed that Buchem et al.\ only mentioned that their algorithm calculates a $(6 + \epsilon)$-approximation for MSD while one can also show a $(4 + \epsilon)$-approximation. We provide a respective proof showing a $(4 + \epsilon)$-approximation guarantee in the connected setting. Since the algorithm also works in the non-connected setting (which is the same as using the complete connectivity graph) this also improves the best known approximation guarantee for unconstrained MSD from $(6 + \epsilon)$ to $(4 + \epsilon)$. 

\subsection{Definitions}

Let us first formally define the clustering problems considered in this work:

\begin{definition}\label{def:conclustering}
In an instance of a connected clustering problem, we are given a set $V$ with $n = |V|$ vertices, a distance metric $d: V\times V \rightarrow \mathbb{R}_{\ge 0}$ on $V$, a positive integer $k\ge 2$, 
and an undirected connectivity graph $G = (V,E)$. 
A feasible solution is a partition of $V$ into $k$ disjoint clusters $P_1, \ldots, P_k$ with corresponding centers $c_1, c_2, \ldots, c_k \in V$ which satisfies that for every $i \in [k]:=\{1,\ldots, k\}$ the subgraph of $G$ induced by $P_i$ is connected and $c_i \in P_i$. The aim is to minimize one of the following objective functions:

\vspace{0.2cm}

\begin{minipage}{0.45\textwidth}
\begin{tabular}{ l l }
 \emph{$k$-center:} & $\displaystyle \max_{i \in [k]} \max_{v \in P_i} d(v,c_i)$ \\[10pt]
 \emph{$k$-median:} & $\displaystyle \sum_{i \in [k]} \sum_{v \in P_i} d(v,c_i)$ \\[14pt]  
 \emph{$k$-means:} &$\displaystyle \sum_{i \in [k]} \sum_{v \in P_i} d(v,c_i)^2$\\
\end{tabular}
\end{minipage}
\begin{minipage}{0.45\textwidth}
\begin{tabular}{ l l }
 \emph{min-sum-radii:} &$\displaystyle\sum_{i \in [k]} \max_{v \in P_i} d(v,c_i)$\\[14pt]
 \emph{min-sum-diameter:} &$\displaystyle\sum_{i \in [k]} \max_{v,w \in P_i} d(v,w)$
\end{tabular}

\end{minipage}

\vspace{0.3cm}

One might note that for MSD the centers are not necessary.
\end{definition} 

%We also define for a connected clustering problem (except MSD) the assignment version of this problem to be the variant where the centers $c_1,\ldots,c_k$ are already given beforehand and one only needs to determine $P_1,\ldots,P_k$.

%In the following we will extend the dynamic programs for connected clustering on trees to also work for a broader class of "treelike" graphs. Formally these classes can be defined over the concept of \emph{tree decompositions} introduced by Robertson and Seymour \cite{robertson1984graph}:

To extend the dynamic programs for connected clustering on tree connectivity graphs we define the concept of \emph{tree decompositions} introduced by Robertson and Seymour \cite{robertson1984graph}:

\begin{definition}[Tree Decomposition]
    A \emph{tree decomposition} of a graph $G = (V,E)$ consists of a tree $T = (V_T,E_T)$ and a corresponding  bag $B_t \subseteq V$ for every $t \in V_t$ such that $\bigcup_{t \in V_T} B_t = V$ and:
    \begin{itemize}
        \item For every $\{u,v\} \in E$, there exists at least one node $t \in V_t$ with $u,v \in B_t$.
        \item For every $v \in V$, the set $S_v = \{t\mid v \in B_t\} $ of nodes whose bags contain $v$ form a connected subtree of $T$.
    \end{itemize}
    The \emph{treewidth} of a tree decomposition $(T,\{B_t\}_{t \in V_T})$ is the size of the largest bag minus one. The treewidth of a graph is the minimum treewidth among all possible tree decompositions of $G$.
\end{definition}

Kloks later introduced the concept of \emph{nice tree decompositions}, which fulfill structural properties that are often desirable \cite{kloks1994treewidth}. In this work we will use an extended version of this concept by Cygan et al.\ that also explicitly deals with the introduction of edges in the decomposition \cite{cygan2011solving}:

\begin{definition}[Nice Tree Decomposition]
    A nice tree decomposition is a tree decomposition with a special node $z$ called the root with $B_z = \emptyset$ and in which each node $t \in V_T$ is one of the following types:
    \begin{itemize}
        \item \textbf{Leaf node:} $t$ is a leaf of $T$ and $B_t = \emptyset$
        \item \textbf{Introduce vertex node:} $t$ has precisely one child $s$ such that $B_t = B_s \cup \{v\}$ for a vertex $v \in V \setminus B_s$. We say that $t$ introduces $v$.
        \item \textbf{Introduce edge node:} A bag node $t$ labeled with an edge $\{u,v\} \in E$ and a single child $s$ such that $B_t = B_s$ and $u,v \in B_t$. We say that $t$ introduces $\{u,v\}$ and every edge is introduced precisely once in our tree decomposition.
        \item \textbf{Forget node:} $t$ has precisely one child $s$ such that $B_t = B_s \setminus \{v\}$ for a vertex $v \in B_s$. We say that $t$ forgets $v$. Every vertex $v$ is forgotten exactly once.
        \item \textbf{Join node:} $t$ has precisely two children $s_1$ and $s_2$ with $B_{s_1} = B_{s_2} = B_t$.
    \end{itemize}
\end{definition}

% A common strategy to solve NP-hard problems is to design algorithms with respect to a specific parameter of the problem like the treewidth. The aim is to design algorithms which run reasonably fast on instances where $p$ is small. We distinguish two complexity classes:
% \begin{itemize}
%     \item An algorithm is slice wise polynomial (XP) with respect to a parameter $p$ if for any fixed parameter $p$ the running time of the algorithm can be bounded by a polynomial in $x$.
%     \item An algorithm is fixed parameter tractable (FPT) with respect to a parameter $p$ if the runtime of the algorithm can be bonded by a function of the form $f(p) x^{c}$ for a computable function $f: \mathbb{N} \rightarrow \mathbb{N}$ and a constant $c$.
% \end{itemize}

It has been shown that if a graph admits a tree decomposition of width $w$ then it also admits a nice tree decomposition of this width. Furthermore, one can always transform a tree decomposition into a nice one of bounded size:

\begin{lemma}[Lemma 7.4 in \cite{cygan2015parameterized}]
    Given a tree decomposition $(T,\{B_t\}_{t \in V_T})$ of a graph $G = (V,E)$ with width $w$, one can compute a nice tree decomposition of $G$ of width $w$ and with $O(w|V|)$ nodes in time $O(\max(|V_T|,|V|)w^2)$.
\end{lemma}

Together with the fact that Korhonen and Lokshtanov proved that one can calculate a tree decomposition of the same treewidth $w$ as the graph $G$ in $O(2^{O(w^2)}|V|^4)$ \cite{korhonen2023improved}, this enables us to calculate a nice tree decomposition of a graph $G$ with minimum width in FPT time. Alternatively there also exist approximation algorithms for tree decompositions \cite{korhonen2022single}.

\subsection{Outline and Our Results}

\begin{table}[h]
\centering
\begin{tabular}{ |c|c|c|c| } 
\hline
Objective & Running Time & Approximation factor & Theorem\\
\hline
$k$-center & $n^{O(w)}$ & $1$ & \ref{thm:tw_center_opt}\\
\hline
$k$-center & $n^2\log(n) + n \log(n) (wk)^{O(w)}$ & $6$ & \ref{thm:tw_center_apx}\\ 
\hline
 $k$-median & $p_1(n) + n \cdot (wk)^{O(w)}$ & $2 \alpha_1 + 2$ & \ref{thm:tw_median}\\ 
\hline
 $k$-means & $p_2(n) + n \cdot (wk)^{O(w)}$ & $8 \alpha_2 + 32$ & \ref{thm:tw_means}\\ 
\hline
\end{tabular}
\caption{Overview over the connected clustering results given a tree decomposition of width $w$. The algorithms for connected $k$-median and connected $k$-means are based on an unconstrained $k$-median algorithm with running time $O(p_1(n))$ and approximation factor $\alpha_1$ and an unconstrained $k$-means algorithm with running time $O(p_2(n))$ and approximation factor $\alpha_2$, respectively.}\label{tab:treewidth}
\end{table}

% \mycomment{
% \begin{table}[h]
% \centering
% \begin{tabular}{ |c|c|c|c| } 
% \hline
% Objective & Running Time & Approximation factor & Theorem\\
% \hline
% $k$-center & XP & $1$ & \ref{thm:tw_center_opt}\\
% \hline
% $k$-center & \multirow{3}{10em}{FPT w.r.t $\max(k,w)$} & $6$ & \ref{thm:tw_center_apx}\\ 
% \cline{1-1}\cline{3-4}
%  $k$-median & & $2 \alpha_1 + 2$ & \ref{thm:tw_median}\\ 
% \cline{1-1}\cline{3-4}
%  $k$-means & &$4 \alpha_2 + 16$ & \ref{thm:tw_means}\\ 
% \hline
% \end{tabular}
% \caption{Overview over the connected clustering results for graphs with bounded treewidth, where $\alpha_1$ and $\alpha_2$ denote approximation guarantees of an arbitrary unconstrained $k$-median and $k$-means approximation algorithm, respectively.}\label{tab:treewidth}
% \end{table}
% }

In Section \ref{sec:techniques} we give an overview over the techniques we used to obtain our results. The techniques are presented in the same order as the results afterwards. 

%Section \ref{sec:tw_general} focuses on our results for graphs with bounded treewidth. An overview over these results can be found in Table \ref{tab:treewidth}.
Section \ref{sec:tw_general} focuses on our results for graphs with bounded treewidth which are summarized in Table \ref{tab:treewidth}. In Section \ref{sec:tw_center_opt} the exact algorithm for connected $k$-center with running time $n^{O(w)}$ is presented. In Section \ref{sec:tw_center_apx} we modify this algorithm to obtain a constant approximation while achieving an FPT running time. 
%Lastly, we present in Section \ref{sec:tw_median} how one can apply the same techniques to obtain approximations for connected $k$-median and $k$-means. While not presented in this work one could also modify the exact connected $k$-center algorithm the same way to obtain exact algorithms for connected $k$-median and $k$-means.
By similar methods as for k-center, one can also obtain exact algorithms for connected k-median and k-means with running time $n^{O(w)}$. In Section \ref{sec:tw_median}, we present how one can apply the same techniques as for k-center to compute approximations for connected k-median and k-means in FPT time.
%One could also modify the exact algorithm the same way to also calculate the optimum solution for the respective problems.

In Section \ref{sec:hardness} we complement these results with a family of instances proving that in the general case one cannot hope for a constant approximation for connected $k$-center unless $P = NP$:

\begin{restatable}{theorem}{ThmHardnessCenter}
    It is NP-hard to approximate connected $k$-center with an approximation ratio in $o(\log^*(k))$. The same is also true if the centers are given in advance.
\end{restatable}

In Section \ref{sec:msr} we show that unlike connected $k$-median and $k$-center, connected min sum radii actually admits a constant approximation:

\begin{restatable}{theorem}{ThmMSR}\label{thm:msr_connected}
    For any $\epsilon > 0$ one can calculate a $(3 + \epsilon)$-approximation for the connected min-sum-radii problem in polynomial time.
\end{restatable}

Additionally we show in Section \ref{sec:msd} that the algorithm also calculates a $(4 + \epsilon)$-approximation for min-sum-diameter which also improves the previously known $(6 + \epsilon)$-approximation guarantee in the unconstrained case:

\begin{restatable}{theorem}{ThmMSD}
    For any $\epsilon > 0$ one can calculate a $(4 + \epsilon)$-approximation for both the connected as well as the regular min-sum-diameter problem in polynomial time.
\end{restatable}

\section{Our techniques} \label{sec:techniques}

\subsection{Results for graphs with bounded treewidth}

\subsubsection{An exact algorithm for connected \texorpdfstring{$k$}{k}-center}

Let a graph $G = (V,E)$ with $n = |V|$ vertices, a natural number $k$ and a nice tree decomposition $\left(T = (V_T,E_T),\{B_t\}_{t \in V_T}\right)$ of $G$ with treewidth $w$ and size $O(wn)$ be given. For any node $t \in V_T$ we define $T_t$ to be the nodes contained in the subtree rooted at $t$ and the graph $G_t = (V_t,E_t)$, where $V_t = \bigcup_{u \in T_t} B_u$ and $E_t = \big\{\{v,u\} \in E \mid \{v,u\} \text{ has been introduced in } T_t\big\}$.

As pointed out by Hochbaum and Shmoys \cite{hochbaum1985best}, every possible objective value of the $k$-center objective corresponds to the distance between a pair of vertices, of which there only exist $O(n^2)$ many. This remains true for the connected $k$-center problem. Thus we can guess the optimal radius $r^*$ by sorting these distances and then applying binary search. For a guessed radius $r$, our goal is then to use dynamic programming to decide whether a connected $k$-center solution with radius $r$ exists. This enables us to find the optimal radius $r^*$ and later also obtain the corresponding solution via backtracking. To apply dynamic programming, we need to be able to calculate for a given node $t$ the number of clusters necessary to partition the vertices in the subgraph $G_t$ independently of the remainder of the graph. We can use the property of the tree decomposition that the vertices in $B_t$ form a vertex separator between $V \setminus V_t$ and $V_t \setminus B_t$. Thus all clusters partially being contained in $V_t$ and $V \setminus V_t$ must be connected via at least one vertex in $B_t$. 

Following the same idea as in the dynamic program for tree connectivity graphs \cite{ge2008joint}, we can fix the centers to which the vertices in $B_t$ get assigned to. Then we know exactly to which centers in $V \setminus V_t$ we could possibly assign the vertices in $V_t$ and vice versa. However this is not sufficient to make the partial solutions for the separated subgraphs entirely independent. If multiple vertices in $B_t$ get assigned to the same center $c$ we know that they must be connected in the final solution. However the connecting path could either pass though $G_t$ or through the remainder of the graph $G$. Thus we may or may not need to assign the vertices in $V_t$ in such a way that the respective vertices in $B_t$ are connected which in turn could separate other clusters in $G_t$ which can influence the number of clusters necessary to cover $G_t$.

%Following the logic of the dynamic program for tree connectivity graphs by Ge et al.~\cite{ge2008joint} one could believe that it would be sufficient to fix the assignments of the vertices in $B_t$ to the respective centers to make the clustering of $V_t$ independent of the remainder of the graph. This is however not the case. This is due to the fact that multiple vertices $v,w$ in $B_t$ could be assigned to the same center which means that in our final solution the respective vertices need to be connected. Depending on the structure of $V \setminus V_t$ it may already be necessary to enforce this connection within $G_t$ because all paths from $v$ to $w$ pass through $V_t$ or it might be possible to connect them via another path in $V \setminus V_t$ (if such a connection is even possible). Depending on this the number of clusters necessary to cover $G_t$ might vary, an example depicting such a situation can be found in Figure~\ref{fig:reason_partition}. 

To circumvent this problem we will not only fix the assignments of the vertices in $B_t$ but also provide a partition which tells us which vertices that get assigned to the same center need to be connected in $G_t$. A similar approach was already used to solve the Steiner tree problem on graphs with bounded treewidth \cite{chimani2012improved}. For every feasible assignment and connection of the vertices, we can then calculate the respective clusters necessary to cover $G_t$. 
For obvious reasons, we only want to consider partitions of $B_t$ that cluster vertices together that are assigned to the same center. More generally:

\begin{restatable}{definition}{SpecificationCenterTW}\label{def:specification}
    Let $S\subseteq V$ be a set of vertices and $a: S \rightarrow V$ be an assignment function on $S$. We say a partition $\mathcal{U}$ of $S$ is \emph{compatible} with $a$ if for any $U \in \mathcal{U}$ and any vertices $v,w \in U$ it holds that $a(v) = a(u)$. With a slight abuse of notation we say for any set $U \in \mathcal{U}$ that $a(U) = a(v)$, where $v$ is an arbitrary vertex in $U$. We call the pair $(a,\U)$ of an assignment function and a compatible partition $\U$ also a \emph{specification} for $S$.
\end{restatable}

%To formalize this idea let $t \in T$ be a node and let $a: B_t \rightarrow S$, where $S$ is an arbitrary set,  be an assignment function of $B_t$. We say a partition $\mathcal{U}$ of $x_t$ is \emph{compatible} with $a$ if for any $U \in \mathcal{U}$ and any vertices $v,w \in U$ it holds that $a(v) = a(u)$. With a slight abuse of notation we say for any set $U \in \mathcal{U}$ that $a(U) = a(v)$ where $v$ is an arbitrary vertex in $U$.

\begin{figure}[t]
\centering
    \includegraphics[width= 0.5 \textwidth]{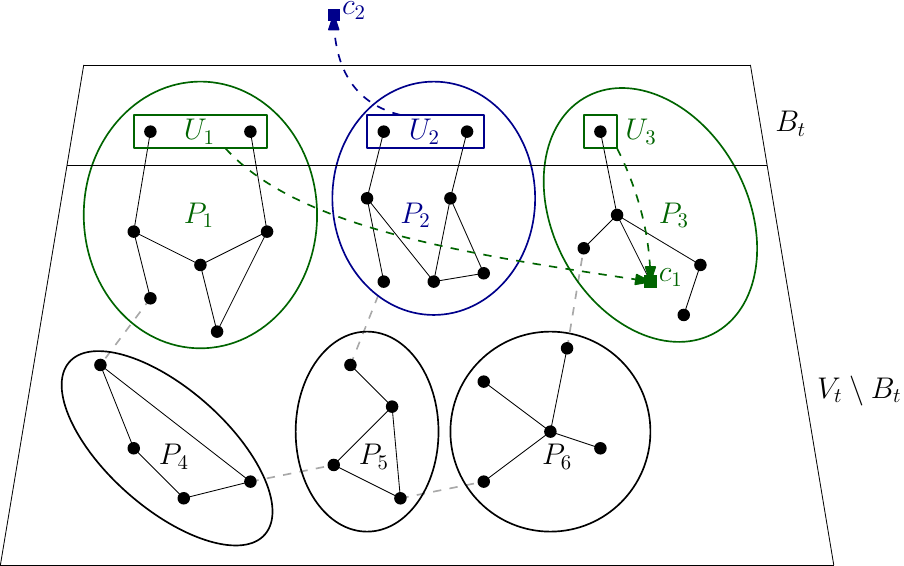}
    \caption{A solution for the graph $G_t$ fulfilling the specification $(a,\{U_1,U_2,U_3\})$ where $a(U_1) = a(U_3) = c_1$ and $a(U_2) = c_2$ with $c_2 \not\in V_t$. The solution requires $3$ clusters, finished ones are circled in black, the unfinished ones in the color corresponding to their assignment.}\label{fig:def_solution_spec}
\end{figure}

Now we would like to calculate for a given $t$ and any specification $(a,\U)$ for $B_t$ the minimum number of clusters necessary to cover $G_t$ if we assign and connect the vertices in $B_t$ accordingly. For technical reasons, we will only count the clusters that are not containing any vertices in $B_t$. We call these clusters \emph{finished} as the other clusters might still continue to grow. For every $U \in \U$ we have to make sure that all the vertices of $U$ are contained in the same cluster already as we require them to be connected in $G_t$. Different sets $U_1, U_2 \in \U$ that are assigned to the same center will however for now be contained in different unfinished clusters that need to be merged in the upper stages of the tree. Furthermore, if there exists a vertex $c \in V_t$ and a vertex $v \in B_t$ with $a(v) = c$, $c$ needs to be contained in an unfinished cluster that is assigned to c (possibly but not necessarily the same cluster that also contains $v$), since otherwise it would not be possible to connect $v$ to $c$. Also for all finished clusters, the respective center must be contained in the cluster itself, while unfinished clusters may not contain the respective center if it is not contained in $G_t$. Lastly, all clusters need to be connected again and each vertex should only have distance at most $r$ to its center. In total we end up with the following definition:

\begin{restatable}{definition}{SolutionCenterTW}\label{def:sol_spec_center}

Let $t \in V_T$ and $(a,\U)$ with $\U = \{U_1,\ldots,U_l\}$ be a specification for $B_t$. We say that a partition $P_1,\ldots,P_l,P_{l+1},\ldots,P_{l + k'}$ of $V_t$ is a connected $k$-center solution for $G_t$ fulfilling $(a,\U)$, if:
\begin{enumerate}
\item For all $i \in [l + k']$ the cluster $P_i$ induces a connected subgraph of $G_t$. \label{prop:sol_connected}
\item For all $i \in [l]$ it holds that $U_i \subseteq P_i$ and for all $v \in P_i$ we have $d(v,a(U_i))\leq r$.\label{prop:sol_dist_uf}
\item For all $c \in V_t$ for which there exists a vertex $v \in B_t$ with $a(v) = c$, there must exist an index $i \in [l]$ such that $a(U_i) = c$ and $c \in P_i$. \label{prop:sol_reachable}
\item For all $i \in \{l+1,\ldots, l + k'\}$ it holds that there exists a center $c \in P_i$ such that for all $v \in P_i$ $d(v,c) \leq r$.\label{prop:sol_dist_finished}
\end{enumerate}
We say that such a solution requires $k'$ clusters and define $p_t(a, \mathcal{U})$ to be the minimum value $k' \in \mathbb{N}_0$ such that there exists a solution for $t$ fulfilling $(a,\U)$ that only requires $k'$ clusters. If such a $k'$ does not exist, we set $p_t( a,\mathcal{U}) = \infty$.

\end{restatable}

Figure~\ref{fig:def_solution_spec} depicts how a solution for a node fulfilling a specification can look like.

Starting with the leaves, our algorithm iterates over all nodes $t \in V_T$ ordered by their height in $T$, calculates the value $p_t(a,\U)$ for any specification $(a,\U)$ of $B_t$ and stores it in a table. As we show in Section~\ref{sec:tw_center_opt}, one can calculate each of these values with a running time bounded by a function only depending on $w$, if we already filled the table for the children of $t$. One may note that for a fixed treewidth $w$ the number of specifications for the bag $B_t$ is polynomially bounded in $n$. As a result, our algorithm runs in polynomial time if the treewidth is a constant. Once we reach the root $z$ we have that $G_z = G$ and $B_z = \emptyset$. 
%Thus there only exists a single specification for $z$ whose corresponding value is exactly equal to the minimum number of centers necessary to cover $G$ with radius $r$ and we can check if this number is at most $k$.
Thus for the unique function $a_\emptyset:\emptyset \rightarrow V$ the value $p_z(a_\emptyset, \emptyset)$ is exactly equal to the minimum number of centers necessary to cover $G$ with radius $r$ and we can check if this number is at most $k$.

%Starting with the leaves, our algorithm iterates over all bag nodes $t \in B_t$ ordered by their height in $T$, calculates the value $p_t(a,\U)$ for any specification $(a,\U)$ of $B_t$ and stores it in a table. As we show in Section~\ref{sec:tw_center_opt}, we can calculate these values in $n^{O(w)}$, if we already calculated them for the children of $t$. One may note that for a fixed treewidth $w$ the number of specifications for the bag $B_t$ is polynomially bounded in $n$. As a result, the resulting algorithm runs in polynomial time on these kind of graphs. Once we reach the root $z$ we have that $G_z = G$ and $B_z = \emptyset$. Thus the value $p_z(\emptyset, \emptyset)$ is exactly equal to the minimum number of centers necessary to cover $G$ with radius $r$ and we can check if this number is at most $k$.

\subsubsection{Obtaining a constant approximation in FPT time}

One might note that the running time of the previous algorithm is not in FPT. This is due to the fact that for a node $t$ there exist up to $n^{|B_t|}$ different assignment functions from $B_t$ to $V$, since for any vertex $v \in B_t$ there are $n$ possible choices where it could be assigned to. Given that the bags can have size $w + 1$, we can have $\Omega(n^{w + 1})$ many specifications for every node. We can reduce this factor by limiting the number of possible center choices in advance. We do this by introducing the connected $k$-center with facilities problem, where we can only choose the centers from a given set of facilities $F$:

\begin{restatable}{definition}{FacilitiesCenterTW}
    In the \emph{connected $k$-center with facilities} problem we are given a graph $G = (V,E)$, a set of facilities $F \subseteq V$, a constant $k$ and a distance metric $d: V^2 \rightarrow \mathbb{R}_{\geq 0}$. The goal is to find a tuple $(c_1,\ldots,c_k) \in F^k$ of $k$ facilities (where one facility might appear multiple times) and  clusters $P_1,\ldots,P_k \subseteq V$ such that:
    \begin{itemize}
        \item $P_1,\ldots,P_k$ is a partition of $V$.
        \item For all $i \in [k]$ the vertices in $P_i$ form a connected subgraph of $G$.
        \item $\max_{i \in [k]}\max_{v \in P_i} d(v,c_i)$ gets minimized.
    \end{itemize}
\end{restatable}
One may note that this definition does not require the clusters to contain the respective centers, which will enable us to choose a facility set independently of the connectivity graph.

As we show in Section \ref{sec:tw_center_apx}, one can modify the dynamic program for connected $k$-center to also calculate the optimal connected $k$-center with facilities solution, while reducing the factor $n^{w+1}$ to $|F|^{w+1}$. If the facility set $F$ does not contain the optimal centers for connected $k$-center, this will obviously not correspond to the optimal connected $k$-center solution. However, if we choose a good unconnected $k$-center solution on the vertices $V$ as $F$ (which can be obtained using a $2$-approximation \cite{hochbaum1985best}), then there exists a facility close to every optimal center, which means that the connected $k$-center with facilities solution is not too much more expensive than the regular connected $k$-center solution. Since the clusters of the connected $k$-center with facilities solution may not contain the respective facilities, we then have to find a feasible center for each cluster to transform the resulting solution into a regular connected $k$-center solution. To do this, one can just take the vertex closest to the center, which only increases the radius by at most $2$. In total we end up with a $6$-approximation while reducing the number of possible assignment functions to $|F|^{w+1} = k^{w +1}$ which results in an FPT running time w.r.t.\ the parameter $\max(k,w)$. The details of this procedure can be found in Section \ref{sec:tw_center_apx}. 
% The overall structure can be summarized by the following three steps:

% \begin{enumerate}
%     \item Calculate centers for unconnected $k$-center via an approximation algorithm $\Rightarrow F$
%     \item Calculate optimum connected $k$-center with facilities solution with facility set $F$.
%     \item Find a feasible center for every resulting cluster to obtain a connected $k$-center solution.
% \end{enumerate}

% \begin{algorithm}
% 	\LinesNumbered
% 	\DontPrintSemicolon
% 	\SetKwInOut{Input}{input}
% 	\SetKwInOut{Output}{output}
%    Calculate centers for unconnected $k$-center via an approximation algorithm $\Rightarrow F$\;
%    Calculate optimum connected $k$-center with facilities solution with facility set $F$.\;
%    Find a feasible center for every resulting cluster\;
%     \caption{FPT approximation algorithm for connected $k$-center.}
% \end{algorithm}

\subsubsection{Results for connected \texorpdfstring{$k$}{k}-median and \texorpdfstring{$k$}{k}-means}

The dynamic program can also be adapted to work for the $k$-median and $k$-means objective. The main difference is that we cannot guess the optimal radius in advance anymore. Following a similar idea as in the work by Eube et al.~\cite{eube2025connectedkmediandisjointnondisjoint}, we extend the specification for the bag $B_t$ of a node $t \in V_T$ to also tell us how many centers we are allowed to place in the graph $G_t$ (by adding a number $k' \leq k$). By doing so, we are able to calculate for every node $t$ and every specification for $B_t$ the cost of the optimal clustering on $G_t$ fulfilling this specification. 

Also for these two objectives, we can use a nesting argument to show that if we choose the centers of a good unconnected $k$-clustering solution as facilities, the cost of the optimal solution only increases by a constant factor if we restrict ourself to these facilities. By calculating the optimal connected clustering with facilites solution and choosing appropriate centers for each cluster afterwards, we can calculate constant approximations for both objectives in FPT time. The details can be found in Section \ref{sec:tw_median}.

\subsection{Hardness of approximation for connected \texorpdfstring{$k$}{k}-center.}

To prove that the connected $k$-center problem is $\Omega(\log^*(k))$-hard to approximate we consider the assignment variant of the problem. Here we are given the centers in advance and only have to decide to which center the vertices get assigned (i.e.\ determine the corresponding clusters). Intuitively one would expect this to make the problem easier and not harder. Indeed we can prove that any $\alpha$-approximation algorithm for the regular variant can be transformed into a $2 \alpha$-approximation algorithm for the assignment variant, meaning that the hardness also caries over if the centers are not given. The hardness is based on a reduction from the 3-SAT problem, in which we have to decide if a given CNF formula $\phi$ can be satisfied. There already exist reductions that show that even with two centers the assignment problem is NP-hard \cite{ge2008joint, drexler2024connected}. However this does only exclude approximation factors better than $3$ as by the triangle inequality the centers cannot be too far apart. We use the respective construction as a gadget in our reduction.

Given that our construction is quite technical, we only present a very broad idea of its structure in this section. The complete proof can be found in Section \ref{sec:hardness} and some of the terms we use in this section (gadget, input, output) are not used in the actual description of the reduction. One might also note that for the asymmetric $k$-center problem there exists an $o(\log^*(n))$-hardness result with a similar bound~\cite{chuzhoy2005asymmetric}. However there are no similarities between the respective hardness constructions. Notably the asymmetric $k$-center problem also becomes polynomially solvable if the centers are given.

To get stronger inapproximability results we use a layer approach. In total we have $L$ layers and our vertices are placed in an $L$-dimensional space, meaning that their positions have $L$ coordinates where each coordinate basically correspond to one of the layers. Using a metric similar to the Manhatten distance, we can enforce that the distance between two vertices is equal to the number of unequal coordinates of their positions\footnote{This is a slight simplification, actually their distance is not equal but proportional to this value}. For every position appearing in our instance, there exists a center that is also placed at this position. However not every vertex placed at that position is directly connected to the center. We construct the instance in such a way that every vertex can be assigned to a center differing in at most one coordinate if $\phi$ can be satisfied while otherwise at least one vertex needs to be assigned to a center that differs in every coordinate, resulting in a gap of $L$ between the radius of these two cases.

Each layer is only connected with the previous and the next layer and the centers are only connected with the first, i.e.\ the outermost layer. Thus paths connecting the vertices in the $L$-th layer to their respective center have to pass through every layer. For a vertex $v$ we say that an assignment does a \emph{mistake} in the $i$-th layer if the assignments in this layer separate $v$ from every center with the same $i$-th coordinate as $v$, meaning that $v$ also needs to be assigned to a center with a different $i$-th coordinate.

The $L$-layer instance is defined iteratively and contains multiple copies of instances with $L-1$ layers. In these copies, we replace the centers by regular vertices which we call the \emph{inputs} of the subinstance. They are the only vertices neighbored with vertices from the first layer of the $L$-layer instance. We also add an additional coordiante at the beginning of the positions of the vertices, which is the same for every vertex in this subinstance. As a result we can also say that we make a mistake for this entire subinstance if we assign all inputs to centers with a
different first coordinate. A $0$-layer instance simply consists of an isolated vertex.

To construct the first layer itself we use gadgets based on the reduction by Ge et al.~\cite{ge2008joint}. Basically one can simulate a Boolean decision by assigning a vertex to one of two centers which then allows us to connect its neighbors to this center while possibly preventing connections to the other center. Following this idea, we can create gadgets modeling the formula $\phi$ that use two centers whose positions only differs in the first coordinate and contains a special set of vertices that we call \emph{outputs}. Each output is supposed to be assigned to one specific center (among the two contained in the gadget). If $\phi$ can be satisfied, we can assign all outputs accordingly, but if $\phi$ cannot be satisfied, at least one output needs to be assigned incorrectly. We create these kind of gadgets for every possible choice of the last $L-1$ coordinates and connect their outputs with the inputs of the instances with $L-1$ layers (where the same output can be used multiple times). The different inputs of the same subinstance are each connected with an output from a different gadget and the connections between the in- and outputs have to follow a specific structure that is not presented in this section. Figure~\ref{fig:overview_network} sketches the overall structure of this construction.

\begin{figure}
\centering
    \includegraphics[width= 0.7 \textwidth]{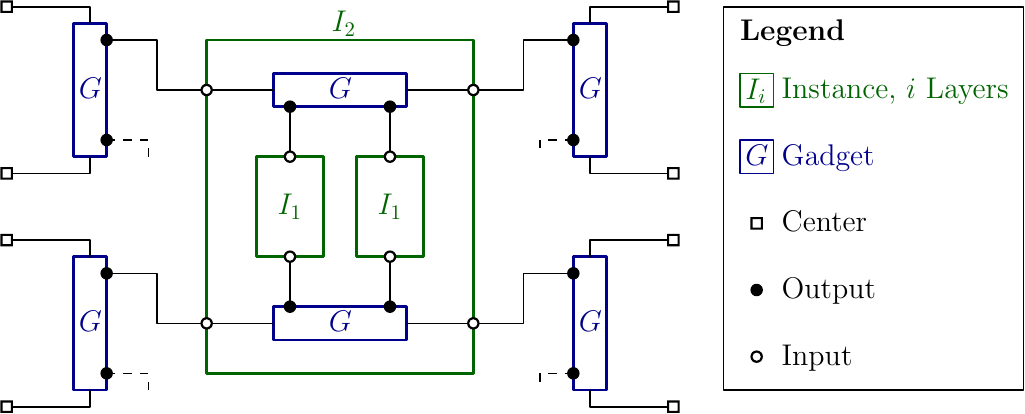}
    \caption{A sketch of a $2$-layer subinstance included in a $3$-layer instance. Note that the number of in- and outputs is not accurate and there are further gadgets, centers and subinstances in the 3-layer instance.}\label{fig:overview_network}
\end{figure}

By adding sufficiently many ($L-1$)-layer instances, one can ensure that if $\phi$ cannot be satisfied (meaning each gadget has an incorrect output), then there also exists at least one of these subinstances which is only connected to incorrect outputs resulting in a mistake for this subinstance. By an inductive argument we can then show that in this $(L-1)$-layer subinstance we have to again make a mistake for one of the contained $(L-2)$-layer subinstances and so on which results in a vertex for which we make a mistake in every layer resulting in a radius of $L$. However if $\phi$ can be satisfied, it is possible to assign the vertices without any mistakes resulting in a radius of $1$.

To make this construction work, we have to cover many different possibilities which outputs are assigned incorrectly by the gadgets, which results in many copies of $(L-1)$-layer subinstances in an $L$-layer instance. Together with the fact that their connections to the gadgets have to fulfill certain requirements, this causes the number of centers to increases exponentially when we add another layer. As a result, the number of layers is bounded by $O(\log^*(k))$ which then also limits our hardness result accordingly.

\subsection{Approximation Algorithms for Connected MSR and MSD}

Our approximation for connected MSR is based on the $(3 + \epsilon)$-approximation algorithm for regular MSR by Buchem et al.~\cite{buchem20243+}. Their algorithm first guesses the $(1/\epsilon)$ largest clusters and then calculates the remaining ones using a primal dual approach. The variables of the primal LP (and thus the constraints of the dual LP) correspond to balls $B(c,r)$ that contain all vertices around a center $c$ with radius $r$. These variables indicate whether or not the respective ball is used as a cluster in the LP solution. To ensure connectivity we have to modify these balls to instead contain all vertices which are reachable from the center via a path of vertices at distance $r$ to it. It is easy to verify that the resulting balls are indeed connected.

Since the original algorithm already merges overlapping balls together, most of the analysis by Buchem et al.\ can be modified to also work with the connectivity constraint. However, one needs to be a little bit careful with the $(1/\epsilon)$ largest clusters that have been guessed in the beginning. In the regular setting, we can just assign all vertices within a certain radius around an optimal center (similarly to the balls) to the center and know that it does not affect the other clusters negatively, if they cannot use the respective vertices anymore. However, in the connected setting, this might split up some of the other optimal clusters. To circumvent this, we allow that the balls in the LP can also contain vertices that are already contained in these guessed clusters and in the end merge them together with the guessed clusters if they overlap. By being careful in the analysis, one can show that this does not increase the approximation factor in the worst case. As a result, we end up with a $(3 + \epsilon)$-approximation algorithm for connected MSR.

To improve the trivial $(6+ \epsilon)$-guarantee that this algorithm gives for the MSD problem (both in the connected and unconnected case), one needs to observe that the cost of the optimal MSD solution is lower bounded by the optimal MSR solution. As a result we can compare against the optimal MSR solution to bound the approximation factor for MSD. In their analysis, Buchem et al.\ were able to prove via an elaborate argument that the radius of a cluster, resulting from merging a specific set of overlapping balls, can be bounded by three times the contribution of these balls to the cost of the dual LP solution. Using an actually far less involved argument, one can show that the diameter of the resulting cluster is bounded by $4$ times the contribution of these balls to the cost of the dual solution. By inserting this bound into the remainder of the proof, we obtain that the MSR algorithm also calculates a $(4 + \epsilon)$-approximation for the optimal MSD solution. This works both in the connected as well as in the regular setting (by using the respective variant of the MSR algorithm).

\section{Algorithms for graphs with bounded treewidth} \label{sec:tw_general}

\subsection{An exact algorithm for connected k-center}\label{sec:tw_center_opt}

Let for a graph $G = (V,E)$, a nice tree decomposition $\left(T,\{B_t\}_{t \in V_T}\right)$ of $G$ with treewidth $w$ and a guessed radius $r$ be given. 
In this section we elaborate how one can calculate for every node $t \in  V_T$ and specification $(a,\U)$ of $B_t$ the value $p_t(a,\U)$ using dynamic programming \footnote{The term specification and the value $p_t(a,\U)$ are defined in Definition \ref{def:specification} and \ref{def:sol_spec_center}}. By doing this we can then verify whether $G$ can be partitioned into $k$ connected clusters with radius $r$ or not. Our algorithm works as follows: We go over all nodes of the nice tree decomposition ordered by their depth in decreasing order. For each node $t$ we have an array $R_t$ that for any specification $(a,\U)$ for $B_t$ contains a value $R_t[a,\mathcal{U}]$ that is initially set to $\infty$. Our goal is to set the values of $R_t$ to the values of $p_t$. If we have already done this for the children of $t$ we can go over all possible specifications $(a',\U')$ for the respective bags and consider how these can be changed by the modifications introduced by $t$ (introduction of a vertex, introduction of an edge, join of two nodes or removal of a vertex) and whether or not they can actually correspond to a feasible solution for the subtree $G_t$. If the latter is the case we will check for the resulting specifications $(a,\mathcal{U})$ for $B_t$ whether the required finished clusters are lower than the current value of $R_t[a,\mathcal{U}]$ or if we have already found a better solution before. If yes the value of $R_t[a,\mathcal{U}]$ is decreased accordingly. The exact details will follow later. After doing this for every entry of the children nodes it holds for all specifications $(a, \mathcal{U})$ for $B_t$ that $R_t[a,\U] = p_t(a,\U)$. Once we reach the root $r$ of the tree decomposition we end up with the numbers of clusters necessary to cover the entire graph while $B_r = \emptyset$, which means that no decisions are left open. Thus we already end up with the minimum number of clusters to cover $G$ with radius $r$. It is sufficient to ensure that the following invariant is fulfilled during the algorithm:

%\begin{definition}
%Let $\U_1$ and $\U_2$ be two partitions of $B_t$. We define $\mathcal{U}_\cup$ to be the partition of $B_t$ fulfilling that for any $U \in \mathcal{U}_1 \cup \mathcal{U}_2$ there exists a $U' \in \mathcal{U}_\cup$ with $U \subseteq U'$ and for any $U' \in \mathcal{U}_{\cup}$ and any partition of $\mathcal{U}_\cup$ into two nonempty subsets $S_1,S_2$ there exists a $U \in \mathcal{U}_1 \cup \mathcal{U}_2$ with $U \cap S_1 \neq \emptyset$ and $U \cap S_2 \neq \emptyset$. Intuitively $\mathcal{U}_{\cup}$ consist thus of the largest sets that can be formed by merging sets in $\mathcal{U}_1 \cup \mathcal{U}_2$ whose intersection is non-empty. We call $\U_\cup$ the \emph{union} of $\U_1$ and $\U_2$ and write $\U_1 \cup \U_2 = \U_\cup$. If for two partitions $\U_1 \cup \U_2 = \U_1$ then we say that $\U_1$ \emph{contains} $\U_2$. 
%\end{definition}

%\begin{definition}
%Let $S$ be a set and $\U_1$ and $\U_2$ two partitions of $S$. We say that $\U_1$ dominates $\U_2$ or $\U_2 \sqsubseteq \U_1$ iff for any $U \in \U_2$ there exists a $U' \in \U_1$ such that $U \subseteq U'$.

%\end{definition}

%\begin{definition}
    %Let $S$ be a set and $\U_1$ and $\U_2$ two partitions of $S$. We call the unique partition $\U_\sqcup$, that dominates both $\U_1$ and $\U_2$ and fulfills that for any $U \in U_\cup$ and any nonempty strict subset $S \subset U$ there exists a $U' \in \U_1 \cup\U_2$ such that $U'$ contains at least one element in $S$ and one in $U \setminus S$, the \emph{combination} of $\U_1$ and $\U_2$ and define $\U_1 \sqcup \U_2 = U_\sqcup$.
%\end{definition}

\begin{invariant}\label{inv:opt_center}
    After $t$ was considered by the algorithm it holds for all specifications $(a,\mathcal{U})$ for $B_t$ that $R_t[a, \U] = p_t(a, \U)$.
\end{invariant}

Now we need to describe how we can calculate $R_t$ for the different kinds of nodes. To simplify notation we sometimes treat the assignment function of a node $t$ as a subset of $B_t \times V$. This way we can modify it by simply adding or removing an assignment.

For a leaf $l$ the set $V_l = B_l$ is empty. Since we clearly do not need any clusters to cover this we can set $R_l(\emptyset, \emptyset) = 0$.

Let $t$ be a node introducing a vertex $v$. Let $s$ be the unique child of $t$. For any specification $(a',\U')$ for $B_t$ and any solution for $G_t$ fulfilling $(a',\U')$ we know that $v$ must form its own cluster as it has no adjacent edges in $G_t$. Thus if $p_t(a',\U') < \infty$ we know that $\{v\} \in \U'$ and that any solution fulfilling $(a',\U')$ can be transformed into a solution for $G_s$ that fulfills $a = a' - (v,a'(v))$ and $\U = \U' - \{v\}$ by simply removing $\{v\}$ which does not increase the number of finished clusters. As a result $p_t(a',\U') \geq p_s(a,\U)$. However we can only transform a solution fulfilling $(a,\U)$ to a solution fulfilling $(a',\U')$ if the assignment of $v$ fulfills certain properties: Obviously we need that $d(v,a'(v)) \leq r$. Additionally if any vertex in $B_s = B_t \setminus \{v\}$ was assigned to $v$ we know by Property~\ref{prop:sol_reachable} that $v$ needs to be assigned to itself. Otherwise $v$ may also be assigned to itself, but it can also be assigned to a center that is not contained in $G_t$ or a center that is contained in $G_t$ to which a vertex in $B_s$ also gets assigned. We will denote the latter set of potential centers as $A(B_s) = \{c \in V \mid \exists u \in B_s: a(u) = c\}$. If we assigned $v$ to a center in $V_s$ to which none of the other vertices in $B_t$ gets assigned Property~\ref{prop:sol_reachable} would again be violated. If the assignment of $v'$ fulfills this properties we know that $p_t(a',\U') = p_s(a,\U)$ as there is a one to one correspondence between solutions fulfilling the respective requirements. Following this logic we can now set the values $R_t[a',\U']$ accordingly using the entries of $R_s$. A more formal description of this procedure can be found in Algorithm~\ref{Alg:IntV_Center_Opt}

\begin{algorithm}
	\LinesNumbered
	\DontPrintSemicolon
	\SetKwInOut{Input}{input}
	\SetKwInOut{Output}{output}
    \textbf{Given:} Node $t$ introducing a vertex $v$  with child $s$\;
    \ForAll{ Specifications $(a,\U)$ for $B_s$}{
    \eIf{ $\exists u \in B_s: a(u) = v$ }{
        $R_t[a + (v,v),\U + \{v\}] = R_s[a,\U]$\;
    }{
    \ForAll{$u \in \left(A(B_s) \cup \{v\} \cup (V \setminus V_t)\right): d(v,u) \leq r$}{
        $R_t[a + (v,u),\U + \{v\}] = R_s[a,\U]$\;
    }
    }
    }
    \caption{Calculating $R$ for an Introduce Vertex node.}\label{Alg:IntV_Center_Opt}
\end{algorithm}

\begin{lemma}
    If for an introduce node $t$ Invariant~\ref{inv:opt_center} was fulfilled for $s$ before $t$ was considered it is also true for $t$ after Algorithm~\ref{Alg:IntV_Center_Opt} was executed on $t$.
\end{lemma}

If the node $t$ introduces an edge $\{v,w\}$ and $s$ denotes again the child of $t$ then for any specification $(a,\U)$ for $B_s$, every solution $\Psc$ for the graph $G_s$ fulfilling $(a,\U)$ is also a solution on $G_t$ fulfilling $(a,\U)$. Additionally the new edge may also be used to connect some clusters $P, P' \in \Psc$. Given that $v,w \in B_t$ we know that both of these clusters need to be unfinished. We can even conclude that if $U_v$ denotes the set in $\U$ containing $v$ and $U_w$ the one containing $w$ that $P$ and $P'$ each need to contain one of these two sets. Obviously such a merge can only happen if $U_v \neq U_w$ and since we want that the clusters are compatible with the assignment function we need to require that $a(U_v) = a(U_w)$. If these conditions are fulfilled however then it is easy to verify that if we merge $P$ and $P'$ in $\Psc$ that the result is a solution for $G_t$ fulfilling $(a,\U')$ where $\U'$ is the partition of $B_t$ if we merge the sets $U_v$ and $U_w$ in $\U$. Since we replaced two unfinished clusters by a single one the number of finished clusters stays the same and thus $p_t(a,\U') \leq p_s(a,\U)$. Following this logic we can iterate over all specifications $(a,\U)$ for $B_s$ and update $R_t[a,\U]$ and $R_t[a,\U']$ (if the sets containing $v$ and $w$ can be merged) accordingly. A more formal description is given in Algorithm~\ref{Alg:IntE_Center_Opt}.

\begin{algorithm}
	\LinesNumbered
	\DontPrintSemicolon
	\SetKwInOut{Input}{input}
	\SetKwInOut{Output}{output}
    \textbf{Given:} Node $t$ introducing a edge $\{v,w\}$  with child $s$\;
    \ForAll{ Specifications $(a,\U)$ for $B_s$}{
    $R_t[a,\U] = \min(R_t[a,\U], R_s[a,\U])$\;
    \If{$a(v) = a(w)$}{
        Let $U_v$ be the set in $\U$ containing $v$ and $U_w$ the set containing $w$\;
        \If{$U_v \neq U_w$}{
            Let $\U'$ be the partition formed by merging $U_v$ and $U_w$ in $\U'$\;
            $R_t[a,\U'] = \min(R_t[a,\U'], R_s[a,\U])$\;
        }
    }
    }
    \caption{Calculating $R$ for an Introduce Edge node.}\label{Alg:IntE_Center_Opt}
\end{algorithm}

\begin{figure}
\centering
    \includegraphics[width= 0.5 \textwidth]{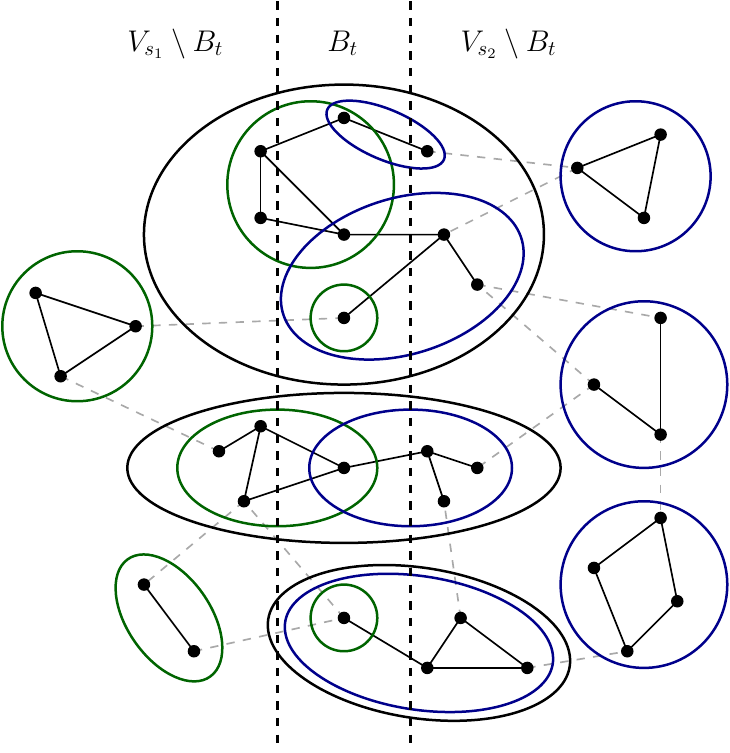}
    \caption{A depictions how a solution for $G_t$ (corresponding to a node $t$ joining two children $s_1$ and $s_2$) can be split up into two solutions for $G_{s_1}$ and $G_{s_2}$. The green ellipsoids correspond to clusters in $G_{s_1}$, the blue ones to clusters in $G_{s_2}$ and the black ones correspond to the unfinished clusters of the solution for $G_t$.}\label{fig:join_multiple_clusters}
\end{figure}

Let $t$ be a node that joins two children $s_1$ and $s_2$. Let $(a,\U)$ be a specification for $B_t$ and $\Psc$ an arbitrary solution fulfilling it. Given that the only common vertices of $G_{s_1}$ and $G_{s_2}$ are in $B_t$ every finished cluster in $\Psc$ is contained in its entirety in one of these two graphs. Additionally, for any unfinished cluster $P \in \Psc$ the vertices in $P \cap B_t$ separate the vertices contained in $V_{s_1}\setminus B_t$ and those in $V_{s_2}\setminus B_t$. $P$ can thus be split up into connected components in $G_{s_1}$ and $G_{s_2}$ that each contain at least one vertex in $B_t$ and thus correspond to unfinished clusters. Figure~\ref{fig:join_multiple_clusters} depicts an example how the clusters of the solution for $G_t$ are split up. Similarly one can split up the set $U \in \U$ that is contained in $P$ into its subsets that are contained in the connected components of $P \cap G_{s_1}$ and $P \cap G_{s_2}$, respectively.
By doing this for every unfinished cluster in $\Psc$ we end up with a partition $\Psc_1$ of $G_{s_1}$, a partition $\Psc_2$ of $G_{s_2}$ (the unfinished clusters are just added to the solution of the graph containing them) and two partitions $\U_1$ and $\U_2$ of $B_t = B_{s_1} = B_{s_2}$ such that $\Psc_1$ is a solution for $G_{s_1}$ fulfilling $(a,\U_1)$ and $\Psc_2$ for $G_{s_2}$ fulfilling $(a,\U_2)$, as we will show later. One may note that these two solutions together contain as many unfinished clusters as $\Psc$. Additionally if we join these two solutions they are connecting the same vertices together as the solution $\Psc$. To be more precise let $a$ and $b$ be two vertices in a cluster $P \in \Psc$. Then the path between them can be split up at the vertices contained in $B_t$ to obtain segments that are in their entirety contained in $G_{s_1}$ or $G_{s_2}$. These segments (including the vertices in $B_t$ surrounding them) are as a result contained in a set in $\Psc_1 \cup \Psc_2$ while consecutive segments contain at least one shared vertex in $B_t$. Thus if we consecutively merge the sets in $\Psc_1 \cup \Psc_2$ with shared vertices we end up with the solution $\Psc$. Since these shared vertices are always contained in $B_t$ the same is true for $\U$ and $\U_1 \cup \U_2$. Because of this we will denote $\U$ as the so called \emph{join} of $\U_1$ and $\U_2$ or more formally:

\begin{definition}
Let $\U_1$ and $\U_2$ be two partitions of $B_t$. We define $\mathcal{U}_\sqcup$ to be the partition of $B_t$ fulfilling that for any $U \in \mathcal{U}_1 \cup \mathcal{U}_2$ there exists a $U' \in \mathcal{U}_\sqcup$ with $U \subseteq U'$ and for any $U' \in \mathcal{U}_{\sqcup}$ and any partition of $U'$ into two nonempty subsets $S_1,S_2$ there exists a $U \in \mathcal{U}_1 \cup \mathcal{U}_2$ with $U \cap S_1 \neq \emptyset$ and $U \cap S_2 \neq \emptyset$. Intuitively $\mathcal{U}_{\sqcup}$ consist thus of the largest sets that can be formed by merging sets in $\mathcal{U}_1 \cup \mathcal{U}_2$ whose intersection is non-empty. We call $\U_\sqcup$ the \emph{join} of $\U_1$ and $\U_2$ and write $\U_1 \sqcup \U_2 = \U_\sqcup$. 
%If for two partitions $\U_1 \sqcup \U_2 = \U_1$ then we say that $\U_1$ \emph{contains} $\U_2$. 
\end{definition}

By our argument above we may follow that for any specifications $(a,\U)$ for $B_t$ we can always find two compatible partitions $\U_1$ and $\U_2$  with $\U_1 \sqcup \U_2 = \U$ such that $p_{s_1}(a,\U_1)+ p_{s_2}(a,\U_2) \leq p_t(a,\U)$.

In the other direction let  $a: B_t \rightarrow V$ be an assignment function and $U_1$ and $U_2$ partitions of $B_t$ that are compatible with $a$. Then $\U = \U_1 \sqcup \U_2 $ is also compatible with $a$ as we only merge sets that contain a common vertex in $B_t$ and are thus are assigned to the same center already. Similarly as above we can also argue that all solutions for $G_{s_1}$ and $G_{s_2}$ that fulfill $(a,\U_1)$ and $(a,\U_2)$, respectively, can be combined to a solution of $G_t$ that fulfills $(a,\U)$ which implies that $p_{s_1}(a,\U_1)+ p_{s_2}(a,\U_2) \geq p_t(a,\U)$. 

Following this logic we can calculate the values of $R_t$ by going over all assignment functions $a: B_t \rightarrow V$ and all pairs of compatible partitions $\U_1,\U_2$ and update the value of $\U_1 \sqcup \U_2$ accordingly (see Algorithm~\ref{Alg:Join_Center_Opt}).

% One may note that since the sets that get merged always contain a common vertex that we are only merging sets assigned to the same center. Thus the resulting partition will still be compatible with the assignment function:

% \begin{observation}
%     For any node $t$, any set $S$ and any assignment function $a: B_t \rightarrow S$ it holds that if two partitions $\U_1$ and $\U_2$ of $B_t$ are compatible with $a$ the same also holds for $\U_1 \sqcup \U_2$.
% \end{observation}

\begin{algorithm}
	\LinesNumbered
	\DontPrintSemicolon
	\SetKwInOut{Input}{input}
	\SetKwInOut{Output}{output}
    \textbf{Given:} Node $t$ with children $s_1$ and $s_2$\;
    \ForAll{ Valid Assignments $a:B_t \rightarrow V$}{
    \ForAll{ Pairs of compatible partitions $\U_1,\U_2$ of $B_t$}{
    $R_t[a, \U_1 \sqcup \U_2] = \min\left(R_t[a, \U_1 \sqcup \U_2], R_{s_1}[a,\U_1] + R_{s_2}[a,\U_2]\right)$\;
    }
    }
    \caption{Calculating $R$ for a Join node.}\label{Alg:Join_Center_Opt}
\end{algorithm}

In the following we give a more formal proof of the correctness of this procedure:

\begin{lemma}\label{lem:join_center}
    If for a join node $t$ Invariant~\ref{inv:opt_center} was fulfilled for $s_1$ and $s_2$ before $t$ was considered it is also true for $t$ after Algorithm~\ref{Alg:Join_Center_Opt} was executed on $t$.
\end{lemma}

%\Jan{Should we maybe only talk about the value of $p$ at this point and talk about the algorithm in a corollary?}

\begin{proof}
    Let $(a,\U)$ with $\U = \{U_1,...,U_j\}$ be a specification for $B_t$. To show that $R_t(a,\U) \leq p_t(a,\U)$ let $\Psc = \{P_1,...,P_{j+k'}\}$ be a solution for $G_t$ fulfilling $(a,\U)$ and requiring $k' = p_t(a,\U)$ clusters. Let $\mathcal{F}^1 = \{F_1^1,...,F_{k_1}^1\}$ be the finished clusters in $P_{j+1},\ldots,P_{j+k'}$ which are subsets of $V_{s_1}$ and $\mathcal{F}^2 = \{F_1^2,...,F_{k_2}^2\}$ the subsets of $V_{s_2}$. One may note that each cluster $P_i$ with $i > j$ is contained in exactly one of these two sets given that $P_i \cap B_t = \emptyset$ and $B_t = V_{s_1} \cap V_{s_2}$. Thus $k_1 + k_2 = k'$.

    For every $i \in [j]$  let $\Psc_i^1 = \{P_{i,1}^1,\ldots,P_{i,l_i}^1\}$ be the connected components of $P_i \cap V_{s_1}$ when we only consider edges in $E_{s_1}$. Similarly we define $\U_i^1 = \{U_{i,1}^1,...,U_{i,l_i}^1\}$ where for any $h \in \{1,\ldots,l_i\}$ we define $U_{i,h}^1 = P_{i,h}^1 \cap B_t$. One may note that $U_{i,h}^1 \neq \emptyset$ as either $P_{i,h}^1= P_i$ or $P_{i,h}^1$ is connected with the remainder of $P_i$ in $G_t$ which means that it had at least on adjacent edge in $E_{s_2}$ which can only happen if it contains a vertex in $B_t$. Furthermore as $\Psc_i^l$ is a partition of $P_i \cap V_{s_1}$ we have that $\U_i^1$ is a partition of $U_i$. Using this argument we may conclude that
    \begin{equation*}
        \U^1= \bigcup_{i = 1,\ldots,j} \U_i^1
    \end{equation*}
    is a partition of $B_t$. Furthermore one can verify that:
    \begin{equation*}
        \Psc^1 = \bigcup_{i = 1,\ldots,j} \Psc_i^1 \cup \mathcal{F}_1
    \end{equation*}
    is a solution for $G_{s_1}$ fulfilling $(a,\U^1)$. The most critical part is that Property~\eqref{prop:sol_reachable} is still fulfilled: Since $\Psc$ fulfills $(a,\U)$ we know that for any $c \in V_t$ with $c = a(u)$ for a vertex $u \in B_t$ there existed an $i \leq j$ with $a(U_i) = c$ and $c \in P_i$. If $c \in V_{s_1} \subseteq V_t$ then there exists a $h$ such that $c \in P_{i,h}^1$ and since $U_{i,h}^1 \subseteq U_i$ and $a(U_i) = c$ we also have that $a(U_{i,h}^1) = c$. At the same time if $c \not\in V_{s_1}$ it does not influence Property~\eqref{prop:sol_reachable} anymore. Thus Property~\eqref{prop:sol_reachable} is fulfilled. The solutions requires exactly $k_1$ clusters as only the clusters in $\mathcal{F}_1$ have no intersection with $B_t$. Thus $p_{s_1}(a,\U^1) \leq k_1$.

    We can also define for $G_{s_2}$ a partition $\U^2$ of $B_t$ accordingly and obtain that $p_{s_2}(a,\U^2) \leq k_2$. Now we will show that $\U^1 \sqcup \U^2 = \U$. Suppose for an $U_i \in \U$ there exists an $S \subset U_i$ such that there exists no set $U' \in \left(\U_i^1 \cup \U_i^1\right)$ with $S \cap U'\neq \emptyset$ and $S \cap U'\neq S$. Then this does also imply that there exists no set $P \in \left(\Psc_i^1 \cup \Psc_i^2\right)$ that contains some vertices in $S$ and some vertices outside of $S$. However this means that neither in  $G_{s_1}$ nor in $G_{s_2}$ there exists a path between vertices in $S$ and vertices in $P_i\setminus S$ (as otherwise the vertices would be added to the same connected component in $\Psc_i^1$ and $\Psc_i^2$, respectively). Since the only common vertices of $G_{s_1}$ and $G_{s_2}$ are in $B_t$ this also means that there exists no path between vertices in $S$ and $P_1 \setminus S$ in $G_t$ which violates that $P_i$ is a connected component in $G_t$. At the same time every set in $\U_i^1 \cup \U_i^2$ is completely contained in $U_i$ while all other sets in $\U^1 \cup \U^2$ are entirely contained in another subset in $\U$. Thus we may conclude that $U_i \in \U^1 \sqcup \U^2$. By applying this logic for any $U_i \in \mathcal{U}$ we obtain $\U^1 \sqcup \U^2 = \U$. Thus $R_t[a,\U]$ is at most $ R_{s_1}[a,\U^1] + R_{s_2}[a,\U^2]$ which by induction is precisely $p_{s_1}(a,\U^1) + p_{s_2}(a,\U^2) $ which is bounded by $k_1 + k_2 = k' = p_t(a,\U)$.

    To show that $R_t[a,\U] \geq p_t(a,\U)$ consider the iteration in which the final value of $R_t[a,\U]$ is set. In this iteration the algorithm considers two partitions $\U_1$ and $\U_2$ of $B_t$ compatible with $a$ such that $\U_1 \sqcup \U_2 = \U$. If we denote $k_1 = R_{s_1}[a,\U_1] = p_{s_1}(a,\U_1)$ and $k_2 = R_{s_2}[a,\U_2]$ then $R_t[a,\U] = k_1 + k_2$. Let $\Psc_1$ and $\Psc_2$ be solutions for $G_{s_1}$ and $G_{s_2}$ that fulfill $(a,\U_1)$ and $(a,\U_2)$, respectively, and require exactly $k_1$ and $k_2$ clusters. One may note that only the vertices in $B_t$ appear both in $\Psc_1$ and $\Psc_2$. Thus the finished clusters of the two solutions are pairwise disjoint. Additionally if we merge the unfinished clusters in $\Psc_1$ and $\Psc_2$ that share a common vertex in $B_t$ we end up with a partition $\Psc$ of $V_t$ where it is easy to verify that every set forms a connected component and that there exists a one to one correspondence between sets $U \in \U$ and the unfinished clusters $P \in \Psc$ that contain the vertices in the respective set $U$ (as in both cases the sets with common vertices in $B_t$ got merged). As a result we obtain that $\Psc$ is a solution for $G_t$ that fulfills $(a,\U)$ and contains exactly $k_1+ k_2$ finished clusters, which implies $p_t(a,\U) \leq k_1 + k_2 = R_t[a,\U]$. Thus the lemma holds. 
\end{proof}

Let $t$ be a Forget node, $v$ the vertex that gets removed and $s$ the child of $t$. We will again consider how a solution for $G_s$ can be transformed into a solution for $G_t$. Let $(a,\U)$ be a specification for $B_s$. Let $U_v \in \U$ be the set containing $v$. We will distinguish whether or not $v$ is the only element in this set. If yes then we know for any solution $\Psc$ for $G_s$ fulfilling $(a,\U)$ that the cluster $P_v \in \Psc$ corresponding to $U_v$ contains no other vertex from $B_s$ and thus will be finished if we forget $v$. Under this circumstances we have to ensure that no other vertex $u$ in $B_t$ is assigned to $a(v)$ as otherwise the cluster containing $u$ and $P_v$ still need to be connected which cannot happen anymore. Additionally we have to make sure that $a(v)$ is contained in $P_v$ as we require for any finished cluster that it contains its center. If no other vertex in $B_t$ is assigned to $a(v)$ this is equivalent to $a(v) \in V_s$ as Property~\ref{prop:sol_reachable} requires that $a(v)$ is contained in one of the open clusters assigned to it if $a(v) \in V_s$. So if both of these requirements are fulfilled we know that $\Psc$ is also a solution for $G_t$ that fulfills $a' = a - (v,a(v))$ and $\U' = \U - \{v\} $ and requires one additional cluster, as $P_v$ gets finished. We thus know $p_t(a',\U') \leq p_s(a,\U) + 1$.

If $U_v$ contains another vertex than $v$ the respective cluster $P_v$ still stays unfinished if we forget $v$. As a result $\Psc$ always forms a solution for $G_t$ that fulfills $(a',\U')$, where in $\U'$ we replace $U_v$ by $U_v - v$. As a result $p_t(a',\U') \leq p_s(a,\U)$.

In a similar fashion we can also transform a solution for $G_t$ into a solution for $G_s$. Thus we can calculate the values of $R_t$ by going over all specifications for $B_s$ and reduce the corresponding entries in $R_t$ accordingly. A more formal description of this procedure can be found in Algorithm~\ref{Alg:Forget_Center_Opt}.

\begin{algorithm}
	\LinesNumbered
	\DontPrintSemicolon
	\SetKwInOut{Input}{input}
	\SetKwInOut{Output}{output}
    \textbf{Given:} Node $t$ forgetting a vertex $v$  with child $s$\;
    \ForAll{ Specifications $(a,\U)$ for $B_s$}{
        Let $U_v$ be the set in $\U$ containing $v$\;
        $a' = a - (v,a(v))$\;
    \eIf{$U_v = \{v\}$}{
        \If{$a(v) \in V_t$ and $U_v$ is the only set in $\U$ assigned to $a(v)$}{
            $\U' = \U - \{v\}$\;
            $R_t\left[a',\U'\right] = \min(R_t\left[a',\U'\right], R_s[a,\U]+1)$\;
        }
        \tcp{Otherwise the specification is not feasible}
    }{
    Define $\U'$ as the partition on $B_t$ when $U_v$ gets replaced by $U_v -v$ in $\U$\;
    $R_t\left[a',\U'\right] = \min(R_t\left[a',\U'\right], R_s[a,\U])$\;
    }
    }
    \caption{Calculating $R$ for a Forget node.}\label{Alg:Forget_Center_Opt}
\end{algorithm}

In the following we present a more detailed proof of the correctness of this procedure:

\begin{lemma}
    If for a forget node $t$ Invariant~\ref{inv:opt_center} was fulfilled for its child $s$ before $t$ was considered it is also true for $t$ after Algorithm~\ref{Alg:Forget_Center_Opt} was executed on $t$.
\end{lemma}

\begin{proof}
    Let $(a',\U')$ be a specification for $B_t$. Let $v$ the vertex that gets forgotten in $t$. First we prove that $R_t[a',\U'] \leq p_t(a',\U')$. Since this is obviously fulfilled if $k' = p_t(a,\U) = \infty$ we only consider the case that $k' \in \mathbb{N}$. Let $\Psc$ be a solution for $G_t$ that fulfills $(a',\U')$ and requires $k'$ clusters. We consider the cluster $P_v \in \Psc$ that contains the vertex $v$. If the cluster is not finished, i.e.\ there exists at least one vertex $u \in B_t \cap P_v$ let $U_v$ be the set in $\U'$ corresponding to $P_v$. If we define $\U$ as the partition on $B_s$ if we add $v$ to $U_v$ and $a:B_s \rightarrow V$ as the assignment function  on $B_s$ where we also set $a(v) = a(U_v)$ it is easy to verify that $\Psc$ is also a solution for $G_s = G_t$ fulfilling $(a,\U)$. Thus by our assumption $R_s[a,\U] = p_s(a,\U) \leq k'$. In the iteration of Algorithm~\ref{Alg:Forget_Center_Opt} in which $a$ and $\U$ get considered the value of $R_t[a',\U']$ gets reduced to at most $R_s[a,\U]$ which in turn ensures that the final value of it will be at most $k' = p_t(a,\U)$.
    
    If $P_v \cap B_t = \emptyset$ let $c_v \in P_v$ be a vertex in $P_v$ such that every vertex in the cluster has distance at most $r$ to it (such a vertex exists because of Property~\eqref{prop:sol_dist_finished}). Let $a: B_s \rightarrow V$ be the function resulting from $a'$ by adding an assignment from $v$ to $c_v$ and let $\U = \U' + \{v\}$. Again we may verify that $\Psc$ is a solution on $G_s$ fulfilling $(a,\U)$. Since the cluster that contains $v$ is not finished when we consider $s$ this solution requires $k' - 1$ clusters. Thus $R_s[a,\U] = p_s[a,\U] \leq k' - 1$. As $c_v$ is contained in $G_s$ and none of the other vertices in $B_t$ are assigned to $c_v$ Algorithm~\ref{Alg:Forget_Center_Opt} thus reduces $R_t[a',\U']$ to $R_s[a,\U] + 1 \leq k' = p_t(a',\U')$ when considering $a$ and $\U$.

    To show that $R_t[a',\U'] \geq p_t(a',\U')$ we consider the moment in which the final value of $R_t[a',\U']$ was set. Let $(a,\U)$ be the specification for $B_c$ considered in this iteration. There are two cases that can happen: If the set $U_v \in \U$ containing $v$ contains further elements then the only difference between $\U'$ and $\U$ is that $v$ gets removed from $\U'$ and the algorithm sets $R_t[a',\U'] = R_s[a,\U]$. It is easy to verify that any solution for $G_s$ fulfilling $(a,\U)$ is also a solution for $G_t$ that fulfills $(a',\U')$. Thus we obtain $R_s[a,\U] = p_s(a,\U) \geq p_t[a',\U']$ which implies also $R_t[a',\U']\geq p_t(a',\U')$.

    If $U_v = \{v\}$ then $\U' = \U - \{v\}$ and the algorithm sets $R_t\left[a',\U'\right] = R_s[a,\U] + 1$ which by our assumption is equal to $p_s(a,\U) + 1$. Furthermore we know that $a(v) \in V_s$ and that no other vertex in $B_t$ is assigned to $a(v)$ since otherwise the algorithm would not have changed the value. Let us consider an arbitrary solution $\Psc$ that fulfills $(a,\U)$ and let $P_v$ be the set containing $v$. Given that $v$ is the only vertex in $B_s$ assigned to $a(v)$ and $a(v) \in V_s$, we know by Property~\eqref{prop:sol_reachable} that $a(v)$ is contained in $P_v$. Thus $P_v$ forms a feasible finished cluster if we forget $v$ and $\Psc$ is a solution for $G_t$ that fulfills $(a',\U')$. However since $P_v$ is now a finished cluster the number of required clusters increases by $1$. Since this holds for all solutions we have that $p_t(a',\U') \leq p_s(a,\U) + 1 = R_s[a,\U] + 1 = R_t\left[a',\U'\right]$. Thus the lemma holds.
\end{proof}

By applying the previous lemmas inductively we obtain that the algorithm calculates the values of $R_t$ correctly for every node $t \in V_T$. We can bound the running time to do this as follows:

\begin{lemma}\label{lem:center_dynamic_r}
    For any node $t$ the values of $R_t$ can be calculated in time $O\left(n^{w + 1}(b_{w + 1})^2p(w)\right)$, where $w$ is the treewidth, $b_{w+1}$ is the $(w+1)$-th.\ Bell-number and $p$ is a polynomial function.
\end{lemma}

\begin{proof}
    We will store the values $B_t$ for any node $t$ as follows: An outer array lists all possible assignment functions $a: B_t \rightarrow V$ and stores for each of them an inner array that stores for each compatible partition (in a fixed order) the respective value $R_t[a,\U]$. The number of assignment functions is trivially upper bounded by $O(n^{w+1})$ as for each of the at most $w + 1$ vertices in $B_t$ we have at most $n$ possible center choices. The number of compatible partitions is clearly upper bounded by the number of all partitions of a set of size $w + 1$ which is by definition exactly the $w + 1$-th.\ Bell number $b_{w + 1}$\footnote{For many connectivity problems on graphs with bounded treewidth it has been shown that one does not actually need to consider the solutions for every possible partition of the bag which reduces the dependency on $w$ \cite{bodlaender2015deterministic}. Given that in our case the number of specifications is dominated by $n^{w+1}$ we did not put any focus on this improvement}. We end up with at most $n^{w + 1} b_{w + 1}$ entries that each can be accessed in constant time.

    If $t$ is a leaf $R_t$ has only a single entry that can be set in time $O(1)$. If it is a Forget, Introduce Vertex or Introduce Edge node we need to iterate over all $O(n^{w + 1} b_{w + 1})$ specifications for $B_s$ (where $s$ is the child of $t$) and only apply some modifications that can be done in time polynomially bounded in $w$.

    For a merge node however we need to iterate for every assignment function $a$ over all pairs of compatible partitions. Given that we store all solutions for a singular assignment function in one of the inner arrays for the children of $t$ this can be done in $O(n^{w + 1}(b_{w + 1})^2)$ For each of this iterations the join of the two partitions can obviously again be calculated in polynomial time and updating the respective entry of $R_t$ costs $O(1)$. Thus there exists a polynomial function $p$ such that for any kind of node the running time is bounded by $O\left(n^{w + 1}(b_{w + 1})^2p(w)\right)$.
\end{proof}

\begin{corollary}
    Given a nice tree decomposition $\left(T,\{B_t\}_{t \in V_T}\right)$ of a graph $G$ with treewidth $w$ and $O(nw)$ nodes, a natural number $k$ and a radius $r$, one can decide whether there exists a connected k-center solution on $G$ with radius $r$ in time $O\left( n^{w + 2}(b_{w + 1})^2p(w)\right)$, where $p$ is a polynomial function.\footnote{The function $p$ differs from the one in Lemma \ref{lem:center_dynamic_r} by a factor $w$ stemming from the number of nodes.}
\end{corollary}

Using binary search we can find the minimum $r$ such that there exists a connected $k$-center solution with radius $r$ by applying our dynamic program $O(\log(n))$ times. To get the actual clustering we can for every node $t \in V_t$ and every entry of $R_t$ store which specifications of the child nodes correspond to the final value of this entry (via a pointer to the corresponding entry of the array $R_s$ of the child $s$). Then we can use backtracking to reconstruct the respective solution. Thus:

\begin{theorem}\label{thm:tw_center_opt}
    Given a nice tree decomposition $\left(T,\{B_t\}_{t \in V_T}\right)$ of a graph $G$ with treewidth $w$ and $O(nw)$ nodes and a natural number $k$, one can calculate an optimal connected k-center solution on $G$ in time $O\left(\log(n)\cdot n^{w + 2}(b_{w + 1})^2p(w)\right)$, where $p$ is a polynomial function.
\end{theorem}

\subsection{Obtaining a constant approximation in FPT time} \label{sec:tw_center_apx}

While the algorithm presented in the previous section has a polynomial running time for a fixed treewidth $w$, its running time is actually not fixed parameter tractable, as the factor $n^{w + 2}$ cannot be bounded by a function purely depending on $w$. Also going away from this specific definition a running time in $\Omega(n^w)$ is clearly not desirable. The aim of this section is to reduce this factor to $n \cdot k^{w+1}$, which is a notable improvement as $k$ tends to be a lot smaller than $n$ and also results in an FPT running time with regard to the parameters $k$ and $w$. However to obtain this speedup we drop the requirement that the algorithm calculates an optimal solution and are content with a $6$-approximation. Given that we provide a reduction that shows for arbitrary graphs that it is NP-hard to approximate connected $k$-center with a constant approximation factor which is independent of $k$, the dependency on the treewidth seems to be reasonable in this context.

We might observe that the factor $n^{ w + 1}$ in our running time analysis results from the fact that for every bag node $t$ we calculate entries of $R_t$ for every possible assignment of the vertices in $B_t$ to the $n$ possible center choices in $V$. To improve this, it would make sense to restrict ourselves to a smaller set $F$ of possible centers. In this context we introduce the \emph{connected $k$-center with facilities} problem:

% \begin{definition}
%     In the \emph{connected $k$-center with facilities} problem we are given a graph $G = (V,E)$, a set of facilities $F$, a constant $k$ and a distance metric $d: (V \cup F)^2 \rightarrow \mathbb{R}_{\geq 0}$. The goal is to find a tuple $(c_1,...,c_k) \in F^k$ of $k$ facilities (where one facility might appear multiple times) and  clusters $P_1,...,P_k \subseteq V$ such that:
%     \begin{itemize}
%         \item $P_1,...,P_k$ is a partition of $V$.
%         \item For all $i \in [k]$ the vertices in $P_i$ form a connected subgraph of $G$.
%         \item $\max_{i \in [k]}\max_{v \in P_i} d(v,c_i)$ gets minimized.
%     \end{itemize}
% \end{definition}

\FacilitiesCenterTW*

One might note that this definition neither requires that the centers of the clusters need to be contained in the cluster itself nor that two different clusters have different centers. In principle the centers are only ensuring that the vertices within one cluster are reasonably close to each other such that we can later find a feasible center that is actually contained in the cluster.

%One might note that $V$ and $F$ are not necessarily disjoint and that we also do not require that all facilities are contained in $V$. Additionally this definition neither requires that the centers of the clusters need to be contained in the cluster itself (which would also be impossible for a facility not contained in $V$) nor that two different clusters have different centers. In principle the centers are only ensuring that the vertices within one cluster are reasonably close to each other such that we can later find a feasible center that is actually contained in the cluster.

By modifying our algorithm for connected $k$-center we are able to calculate the optimal connected $k$-center with facilities solution for given tree decompositions. The factor $n^{w + 1}$ changes into $|F|^{w+1}$, the exact details will be presented below.

\begin{theorem}\label{thm:con_center_facilities}
    Given a nice tree decomposition $\left(T,\{B_t\}_{t \in V_T}\right)$ of a graph $G$ with treewidth $w$ and $O(nw)$ nodes, a natural number $k$ and a facility set $F$, one can calculate the optimal connected k-center with facilities solution in time $O\left(n^2\log(n) +\log(n)\cdot n\cdot |F|^{w + 1}(b_{w+1})^2p(w)\right)$, where $p$ is a polynomial function.
\end{theorem}

The aim is now to transform our connected $k$-center instance into an instance with facilities by finding a small facility set such that if we limit ourselves to this facility set, the cost of the optimal solution does not increase by too much. This is the case if for every vertex (including the optimal centers of the connected $k$-center clustering) there exists a nearby facility:

\begin{lemma}\label{lem:pos_nesting_center}
    If for every $v \in V$ there exists a facility $f \in F$ at distance $r'$, the radius of the optimal connected $k$-center with facilities solution is upper bounded by the radius of the optimal connected $k$-center solution plus $r'$.
\end{lemma}

\begin{proof}
    Let $P_1^*,\ldots,P_k^*$ be the optimal connected $k$-center solution with centers $c_1^*,...,c_k^*$. For any $i \in [k]$ let $f_i^*$ be the facility in $F$ minimizing $d(c_i^*,f_i^*)$. Then the clusters $P_1^*,\ldots,P_k^*$ with centers $f_1^*,...,f_k^*$ form a connected $k$-center with facilities solution with radius:
    \begin{equation*}
        \max_{i \in [k]}\max_{v \in P_i^*} d(v,f_i^*) \leq \max_{i \in [k]}\max_{v \in P_i^*} d(v,c_i^*) + d(c_i^*,f_i^*) \leq \max_{i \in [k]}\max_{v \in P_i^*} d(v,c_i^*) + r'
    \end{equation*}
\end{proof}

One might note that the requirement in Lemma~\ref{lem:pos_nesting_center} basically is the same as saying that the facility set $|F|$ of size $l$ is an $l$-center solution with radius $r'$ that ignores the connectivity requirement. Following this logic, we can calculate an unconstrained $k$-center solution and use this as our facility set:

\begin{lemma}\label{lem:facilities_center}
    For any connected $k$-center instance with optimal radius $r^*$, one can calculate in $O(n^2 \log(n))$ time a facility set $F \subseteq V$ with $|F| = k$ such that there exists a solution for the connected $k$-center with facilities problem using $F$ with radius at most $3r^*$.
\end{lemma}

\begin{proof}
    By ignoring the connectivity constraint one can use the classic algorithm by Hochbaum and Shmoys \cite{hochbaum1985best} to calculate a 2-approximation for $k$-center on the vertex set $V$ with metric $d$. Let $c_1,...,c_k$ be the respective centers. Given that removing the connectivity constraint only relaxes the problem we know that the radius of this solution is upper bounded by $2r^*$. Thus Lemma~\ref{lem:pos_nesting_center} tells us that  if we choose $F = \{c_1,...,c_k\}$ the radius of the optimal connected $k$-center with facilities solution is upper bounded by $3r^*$.
\end{proof}

As we argue below one can calculate the optimal solution of the resulting connected $k$-centers with facilities instance with a similar algorithm as for the regular connected $k$-center problem. The factor $n^{w + 1}$ decreases to $|F|^{w + 1}$ which in our case is equal to $k^{w + 1}$. In the last step we need to find a feasible center for each of the resulting clusters. One might observe that by taking the vertex closest to the corresponding facility of the cluster the radius of the clustering increases by at most a factor of two:

\begin{observation}\label{obs:cen_rem_facilities}
    Every connected $k$-center with facilities solution can be transformed into a connected $k$-center solution while increasing the radius by a factor of at most $2$.
\end{observation}

In total we can use the algorithm for connected $k$-center with facilities to obtain a $6$-approximation for the connected $k$-center problem:

\begin{theorem}\label{thm:tw_center_apx}
    Given a nice tree decomposition of a graph $G$ with treewidth $w$ and $O(nw)$ bag nodes and a natural number $k$, one can calculate a 6-approximation of connected $k$-center in time $O\left(n^2\log(n) + n \log(n) (wk)^{O(w)}\right)$, which lies in FPT with respect to the parameter $\max(w,k)$.
\end{theorem}

\begin{proof}
    Let $r^*$ be the radius of the optimal connected $k$-center solution. By Lemma~\ref{lem:facilities_center} we can calculate in $O(n^2 \log(n))$ time a facility set of size $k$ such that the optimum solution for the resulting connected $k$-center with facilities instance has at most radius $3r^*$. By Theorem~\ref{thm:con_center_facilities} we can calculate this solution in time $O\left(n^2\log(n) +\log(n)\cdot nw\cdot k^{w+1}(b_{w+1})^2p(w)\right)$, which which lies in $O\left(n^2\log(n) + n \log(n) (wk)^{O(w)}\right)$. Lastly by Observation~\ref{obs:cen_rem_facilities} we can find for each cluster a center that is contained in it while increasing the radius by at most $2$. Thus we obtain a solution with radius at most $2\cdot 3r^* = 6r^*$.
\end{proof}

It remains to explain how we can calculate a solution for the connected $k$-center with facilities problem. The approach is rather similar to the dynamic program in the previous section. For a guessed radius $r$ we calculate a value $R_t[a,\U]$ for any bag node $t$ and any specification $(a,\U)$ of $B_t$ that tells us how many clusters we need, to partition $G_t$ under this specification. The main difference is that the assignment function $a:B_t \rightarrow F$ only assigns vertices to the facility set $F$. Additionally we do not require anymore that the centers of the clusters are contained in the clusters themselves. As a result we do not need Property~\ref{prop:sol_reachable} of Definition~\ref{def:sol_spec_center} anymore when defining a solution for $G_t$ fulfilling $(a,\U)$:

\begin{definition}

Let $t \in T$ be a bag node and $(a,\U)$ with $a: B_t \rightarrow F$ and $\U = \{U_1,\ldots,U_l\}$ a specification for $B_t$. We say that a partition $P_1,\ldots,P_l,P_{l+1},\ldots,P_{l + k'}$ of $V_t$ is a connected $k$-center with facilities solution for $G_t$ fulfilling $(a,\U)$ if:
\begin{enumerate}
\item For all $i \in \{1,\ldots,l + k'\}$ the cluster $P_i$ forms a connected component of $G_t$. 
\item For all $i \in \{1,\ldots,l\}$ it holds that $U_i \subseteq P_i$ and for all $v \in P_i$ we have $d(v,a(U_i))\leq r$.
\item For all $i \in \{l+1,\ldots, l + k'\}$ it holds that there exists an $f \in F$ such that for all $v \in P_i$ $d(v,f) \leq r$.
\end{enumerate}
We say that such a solution requires $k'$ clusters and define $p_t(a, \mathcal{U})$ to be the minimum value $k' \in \mathbb{N}$ such that there exists a solution for $t$ fulfilling $(a,\U)$ that only requires $k'$ clusters. If such a $k'$ does not exist, we say $p_t( a,\mathcal{U}) = \infty$.
\end{definition}

Starting with the leaf nodes one can recursively calculate for every bag node $t$ the value $p_t(a,\U)$ and store it in $R_t(a,\U)$ for any specification $(a,\U)$. For Introduce Edge nodes and Join nodes one can follow the exact same procedure as in Algorithm~\ref{Alg:IntE_Center_Opt} and~\ref{Alg:Join_Center_Opt}. The primary function of both algorithms is that they merge sets in $\U$ if they are assigned to the same center and a new connection between them is available (either via a new edge or a path in the subgraph of the other child of the join node). Given that for every vertex it is already determined where it gets assigned the change in the definition of valid solutions does not affect these procedures besides the fact that the number of assignment functions decreases.

If $t$ is introducing a vertex $v$ we do only consider assignments of $v$ to the facilities in $F$. Additionally since we do not require anymore that the centers themselves are contained in the clusters and thus also need not to ensure, that the respective facilities are reachable, we can assign $v$ without checking where the other vertices in $B_t$ got assigned to, which simplifies the algorithm. The resulting procedure can be found in Algorithm~\ref{Alg:IntV_Center_Apx}.

\begin{algorithm}
	\LinesNumbered
	\DontPrintSemicolon
	\SetKwInOut{Input}{input}
	\SetKwInOut{Output}{output}
    \textbf{Given:} Node $t$ introducing a vertex $v$  with child $s$\;
    \ForAll{ Specifications $(a,\U)$ for $B_s$}{
    \ForAll{$f \in F: d(v,f) \leq r$}{
        $R_t[a + (v,f),\U + \{v\}] = R_s[a,\U]$\;
    }
    }
    \caption{Calculating $R$ for an Introduce Vertex node.}\label{Alg:IntV_Center_Apx}
\end{algorithm}

The fact that the clusters do not need to contain the respective facility also affects the algorithm when $t$ is a Forget node, deleting a vertex $v$. In the algorithm for connected $k$-center one needs to check, whether removing the center finishes the cluster. If yes one has to make sure that no other vertex in $B_t$ is assigned to this center as otherwise the two clusters would still need to be connected in the remainder in the graph which is not possible anymore. However, when we introduce facilities two distinct clusters can have the same center. Thus if the cluster containing $v$ gets finished, one can simply increase the number of required clusters (which are the finished ones) by one independently of the assignment of the other vertices in the bag. For the same reasons one also does not need to check whether the center to which $v$ gets assigned is contained in the subgraph $G_t$. As a result we end up with Algorithm~\ref{Alg:Forget_Center_Apx}.

\begin{algorithm}
	\LinesNumbered
	\DontPrintSemicolon
	\SetKwInOut{Input}{input}
	\SetKwInOut{Output}{output}
    \textbf{Given:} Node $t$ forgetting a vertex $v$  with child $s$\;
    \ForAll{ Specifications $(a,\U)$ for $B_s$}{
        Let $U_v$ be the set in $\U$ containing $v$\;
        $a' = a - (v,a(v))$\;
    \eIf{$U_v = \{v\}$}{
            $\U' = \U - \{v\}$\;
            $R_t\left[a',\U'\right] = \min(R_t\left[a',\U'\right], R_s[a,\U]+1)$\;
    }{
    Define $\U'$ as the partition on $B_t$ when $U_v$ gets replaced by $U_v -v$ in $\U$\;
    $R_t\left[a',\U'\right] = \min(R_t\left[a',\U'\right], R_s[a,\U])$\;
    }
    }
    \caption{Calculating $R$ for a Forget node.}\label{Alg:Forget_Center_Apx}
\end{algorithm}

The correctness of this algorithm can be shown analogously as for the dynamic program for the connected $k$-center problem. Regarding the running time the main difference is that the maximum number of specifications of a bag node $t$ can be bounded by $|F|^{w + 1} b_{w+1}$ instead of $n^{w+1} b_{w+1}$, as we are considering assignment functions from $B_t$ to $F$ instead of $V$. This proves Theorem~\ref{thm:con_center_facilities}.

\subsection{An FPT approximation algorithm for connected \texorpdfstring{$k$}{k}-median and \texorpdfstring{$k$}{k}-means} \label{sec:tw_median}

%\Jan{Add definition with facilities?}

In the following we present how our approach can be used to approximate $k$-median on graphs with bounded treewidth in FPT time. The idea is to combine it with the dynamic program for $k$-median on trees by Eube et al.~\cite{eube2025connectedkmediandisjointnondisjoint}. In the $k$-median problem the distance of all vertices to the respective center contributes to the objective function. As a result we cannot guess a radius $r$ in advance and determine how many clusters are needed to cover a subgraph of $G$ with this radius. Instead we have to add a value $k' \leq k$ to our specification for a node $t$ that tells us how many finished clusters are allowed within $G_t$ and determine the minimum cost required to cover $G_t$ under these circumstances.

To avoid repetition, we do not present how connected $k$-median can be solved exactly on graphs with bounded treewidth but concentrate on the dynamic program for connected $k$-median with facilities which can be defined analogously to the $k$-center variant. The modifications necessary are roughly the same. Afterwards we present how one can obtain a constant approximation for connected $k$-median in FPT time with our algorithm. Let $t \in V_T$. For the connected $k$-median with facilities problem a specification $(k',a,\U)$ for $B_t$ consists of a natural number $k'$ with $0 \leq k' \leq k$, an assignment function $a:B_t \rightarrow F$ and a compatible partition $\U$ of $B_t$. We then define the value $p_t(k',a,\U)$ as the minimum cost of a solution fulfilling the specification. For any vertex $v \in B_t$ we ignore the distance from $v$ to its facility $a(v)$ and only add this distance once the vertex is forgotten. This avoids that we count the distance multiple times in a join node:

\begin{definition}

Let $t \in T$ be a bag node and $(k',a,\U)$ with $k' \leq k$, $a: B_t \rightarrow F$ and $\U = \{U_1,\ldots,U_l\}$ a specification for $B_t$. We say that a partition $P_1,\ldots,P_l,P_{l+1},\ldots,P_{l + k'}$ is a connected $k$-median with facilities solution for $G_t$ fulfilling $(k',a,\U)$ if:
\begin{enumerate}
\item For all $i \in [l + k']$ the cluster $P_i$ forms a connected component of $G_t$. 
\item For all $i \in [l]$ it holds that $U_i \subseteq P_i$.
\end{enumerate}
We say that such a solution costs
\begin{equation*}
    \sum_{i = 1}^l\left(\sum_{v \in P_i \setminus U_i} d(v,a(U_i))\right) +  \sum_{i= l + 1}^{l + k'} \left(\min_{f \in F} \sum_{v \in P_{i}} d(v,f)\right)
\end{equation*}
 and define $p_t(k', a, \mathcal{U})$ to be the minimum cost of a solution for $G_t$ fulfilling $(k', a, \mathcal{U})$. If there is no solution fulfilling $(k',a,\U)$ we say $p_t( a,\mathcal{U}) = \infty$.

\end{definition}

As for the connected $k$-center problem the algorithm maintains for every bag node $t$ an array $R_t$ that stores for any specification $(k',a,\U)$ for $B_t$ a value $R_t[k',a,\U]$. Initially all of these values are set to $\infty$. Starting from the leaves the algorithm goes over all bag nodes $t \in V_T$ ordered by their height and recursively calculates the value $p_t(k',a,\U)$ for every specification for  $B_t$ and stores it in $R_t[k',a,\U]$. Once we reach the root $z$ with $B_z = \emptyset$, the value $R[k,\emptyset,\emptyset]$ is equal to the cost of the optimal connected $k$-median with facilities solution.

For a leave node $l$ we have $B_l = \emptyset$ and can set $B_l[0,\emptyset,\emptyset] = 0$ and $B_l[k',\emptyset,\emptyset] = \infty$ for $k' \neq 0$ as there is no partition of the emptyset of size greater $0$.

Let $t$ be a introduce vertex node, introducing $v \in V$ and let $s$ be the child of $t$. As $v$ is an isolated vertex in $G_t$ it cannot be connected with any other vertex in $B_t$. Given that we only count finished clusters and ignore the distances of the vertices in $B_t$ to the respective facilities  we can set $R_t[k',a + (v,f),\U + \{v\}] = R_s[k',a,\U]$ for any specification $(k',a,\U)$ for $B_s$. This results in Algorithm~\ref{Alg:IntV_Median_Apx}.

\begin{algorithm}
	\LinesNumbered
	\DontPrintSemicolon
	\SetKwInOut{Input}{input}
	\SetKwInOut{Output}{output}
    \textbf{Given:} Node $t$ introducing a vertex $v$  with child $s$\;
    \ForAll{ Specifications $(k',a,\U)$ for $B_s$}{
    \ForAll{$f \in F$}{
        $R_t[k',a + (v,f),\U + \{v\}] = R_s[k',a,\U]$\;
    }
    }
    \caption{Calculating $R$ for an Introduce Vertex node.}\label{Alg:IntV_Median_Apx}
\end{algorithm}

If the bag node $t$ with child $s$ introduces an edge $\{v,w\}$, every solution for $G_s$ is also a solution for $G_t$ fulfilling the same specifications. Additionally the edge could connect two unfinished clusters if $a(v) = a(w)$. Following this logic, we end up with Algorithm~\ref{Alg:IntE_Median_Apx}.

\begin{algorithm}
	\LinesNumbered
	\DontPrintSemicolon
	\SetKwInOut{Input}{input}
	\SetKwInOut{Output}{output}
    \textbf{Given:} Node $t$ introducing a edge $\{v,w\}$  with child $s$\;
    \ForAll{ Specifications $(k',a,\U)$ for $B_s$}{
    $R_t[k',a,\U] = \min(R_t[k',a,\U], R_s[k',a,\U])$\;
    \If{$a(v) = a(w)$}{
        Let $U_v$ be the set in $\U$ containing $v$ and $U_w$ the set containing $w$\;
        \If{$U_v \neq U_w$}{
            Let $\U'$ be the partition formed by merging $U_v$ and $U_w$ in $\U$\;
            $R_t[k',a,\U'] = \min(R_t[k',a,\U'], R_s[k',a,\U])$\;
        }
    }
    }
    \caption{Calculating $R$ for an Introduce Edge node.}\label{Alg:IntE_Median_Apx}
\end{algorithm}

If the node $t$ joins two nodes $s_1$ and $s_2$ with $B_t = B_{s_1} = B_{s_2}$ let $(k',a,\U)$ be a specification of $B_t$. We know that in a solution of $G_t$ fulfilling  $(k',a,\U)$ the finished clusters are in their entirety contained in $V_{s_1} \setminus B_t$ or $V_{s_2} \setminus B_t$. Additionally we also know that only vertices in $V_{s_1} \setminus B_t$ or $V_{s_2} \setminus B_t$ contribute to the cost of the solution. Let $\Psc$ be the optimal solution of $G_t$ fulfilling  $(k',a,\U)$. Following the same logic as for connected $k$-center we know that there exist specifications $(k_1,a,\U_1)$ for $B_{s_1}$ and $(k_2,a,\U_2)$ for $B_{s_2}$ such that $k_1 + k_2 = k'$ and $\U_1 \sqcup \U_2 = \U$ and $\Psc$ can be subdivided into two solutions for $G_{s_1}$ and $G_{s_2}$ fulfilling the respective specifications (see Lemma~\ref{lem:join_center}). The sum of the cost of these two  solutions is then exactly the cost of $\Psc$ itself, as every vertex is still assigned to the same facility and thus $p_t(k',a,\U) = p_{s_1}(k_1,a,\U_1) + p_{s_2}(k_2,a,\U_2)$. 

At the same time it holds for any pair of specifications $(k_1,a,\U_1)$ for $B_{s_1}$ and $(k_2,a,\U_2)$ for $B_{s_2}$ that two solutions for $G_{s_1}$ and $G_{s_2}$ fulfilling these specifications can be combined to a solution for $G_t$ fulfilling $(k_1 + k_2,a,\U_1 \sqcup \U_2)$, which implies $p_t(k_1 + k_2,a,\U_1 \sqcup \U_2) \leq p_t(k_1,a,\U_1) + p_t(k_2,a,\U_2)$. As a result we can simply go over all pairs of specifications of the child nodes with the same assignment function and update the entry of the join of these two specifications accordingly. Algorithm~\ref{Alg:Join_Median_Apx} gives a more formal description of this procedure.

\begin{algorithm}
	\LinesNumbered
	\DontPrintSemicolon
	\SetKwInOut{Input}{input}
	\SetKwInOut{Output}{output}
    \textbf{Given:} Node $t$ with children $s_1$ and $s_2$\;
    \ForAll{ Assignments $a:B_t \rightarrow F$}{
    \ForAll{$k_1,k_2: k_1 + k_2 \leq k$}{
    $k' = k_1 + k_2$\;
    \ForAll{ Pairs of compatible partitions $\U_1,\U_2$ of $B_t$}{
    $R_t[k',a, \U_1 \sqcup \U_2] = \min\left(R_t[k',a, \U_1 \sqcup \U_2], R_{s_1}[k_1,a,\U_1] + R_{s_2}[k_2,a,\U_2]\right)$\;
    }
    }
    }
    \caption{Calculating $R$ for a Join node.}\label{Alg:Join_Median_Apx}
\end{algorithm}

Let $t$ be a forget node with child $s$ that forgets $v$. Given that $v \in B_s$ and $v \not\in B_t$ we have to add the distance from $v$ to $a(v)$ to the cost when we transform a solution for $G_s$ into one for $G_t$. Additionally we have to check whether $v$ is the only vertex in $B_s$ that is contained in its respective cluster, which is the same as checking if $\{v\} \in \U$ for a given specification $(k',a,\U)$ for $B_s$. If yes the respective cluster turns into a finished one when we forget $v$. Accordingly we have to increase the value of $k'$ by $1$. This results in Algorithm~\ref{Alg:Forget_Median_Apx}.

\begin{algorithm}
	\LinesNumbered
	\DontPrintSemicolon
	\SetKwInOut{Input}{input}
	\SetKwInOut{Output}{output}
    \textbf{Given:} Node $t$ forgetting a vertex $v$  with child $s$\;
    \ForAll{ Specifications $(k',a,\U)$ for $B_s$}{
        Let $U_v$ be the set in $\U$ containing $v$\;
        $a' = a - (v,a(v))$\;
    \eIf{$U_v = \{v\}$}{
            $\U' = \U - \{v\}$\;
            $R_t\left[k' + 1,a',\U'\right] = \min(R_t\left[k' + 1,a',\U'\right], R_s[k',a,\U]+d(v,a(v)))$\;
    }{
    Define $\U'$ as the partition on $B_t$ when $U_v$ gets replaced by $U_v -v$ in $\U$\;
    $R_t\left[k',a',\U'\right] = \min(R_t\left[k',a',\U'\right], R_s[k',a,\U] +d(v,a(v)))$\;
    }
    }
    \caption{Calculating $R$ for a Forget node.}\label{Alg:Forget_Median_Apx}
\end{algorithm}

%By iteratively applying these algorithms on all nodes of $T$ we end up with the optimum connected $k$-median with facilities solution.

\begin{theorem}\label{thm:con_median_facilities}
    Given a nice tree decomposition $\left(T,\{B_t\}_{t \in V_T}\right)$ of a graph $G= (V,E)$ with tree-width $w$ and $O(nw)$ nodes, a natural number $k$ and a facility set $F$, one can calculate the optimal connected k-median with facilities solution in time $ \left(n\cdot k^2\cdot |F|^{w + 1}(b_{w + 1})^2p(w)\right)$, where $p$ is a polynomial function.
\end{theorem}

\begin{proof}
    The correctness of the dynamic program can be shown the same way as for the $k$-center variant. Regarding the running time the number of specifications for the bag $B_t$ of a node $t$ can be bounded by $(k + 1) \cdot |F|^{|B_t|}\cdot b_{|B_t|} \leq (k + 1) \cdot |F|^{w + 1}\cdot b_{w + 1}$. For an introduce vertex, introduce edge or forget node we only need time polynomial in $w$ for every specification of the child node. For the join node we have to combine for every assignment function $a: B_t \rightarrow F$ the results for every pair of compatible partitions (which are at most $(b_{w + 1})^2$) for the bags of the child nodes and choices of values $k_1, k_2$ with $k_1 + k_2 \leq k$ (of which there are $O(k^2)$ many). As a result the running time of a join node is bounded in $O(k^2\cdot |F|^{w + 1}(b_{w + 1})^2 p(w))$, which is the largest among any type of nodes. Since there are only $O(nw)$ bag nodes the theorem holds.
\end{proof}

To obtain a constant approximation for connected $k$-median we have to find a facility set $F$ such that restricting us to this facility set does not increase the cost of the optimal solution by too much. As for connected $k$-center this is the case if the cost of a non-connected $k$-median clustering solution with center set $F$ is reasonably small. To be more precise for a graph $G = (V,E)$, a distance metric $d$ and a set $F$ we define
\begin{equation*}
    cost(F) = \sum_{v \in  V} \min_{f \in F} d(v,f),
\end{equation*}
which is exactly the k-median objective if we use $F$ as a center set. By the same arguments as used in the nesting technique for hierarchical clustering \cite{lin2010general} we can bound the cost of the optimal connected $k$-median with facilities solution by $cost(F)$ plus twice the cost of the optimal solution without facilities:

\begin{lemma}\label{lem:facility_nesting_median}
    For a graph $G = (V,E)$ with distance metric $d$ and an integer $k \in \mathbb{N}$, let $\opt$ be the cost of the optimal connected $k$-median solution. For any set of facilities $F$ we have that the cost of the optimal $k$-median with facilities solution is upper bounded by $cost(F) + 2 \cdot \opt$.
\end{lemma}

\begin{proof}
    Let the clusters $P_1^*$,\ldots,$P_k^*$ with centers $c_1^*,\ldots,c_k^*$ be a connected $k$-median solution with cost $\opt$. For every $i \in [k]$ we set $f_i^* = \argmin_{f \in F} d(f,c_i^*)$. Let $v \in P_i^*$, we would like to bound $d(v,f_i^*)$. Let $f_v$ be the facility in $F$ closest to $v$. Then it holds that $d(c_i^*,f_i^*) \leq d(c_i^*,f_v)$. By using the triangle inequality we obtain $d(c_i^*,f_i^*) \leq d(c_i^*,v) + d(v,f_v)$. Additionally $d(v,f_i^*) \leq d(v,c_i^*) + d(c_i^*,f_i^*)$ and thus:
    \begin{equation*}
        d(v,f_i^*) \leq  d(v,f_v) + 2 d(v,c_i^*)
    \end{equation*}

    Thus the clusters $P_1^*$,\ldots,$P_k^*$ and facilities $f_1^*,\ldots,f_k^*$ form a connected $k$-median with facilities solution with cost
    \begin{align*}
        \sum_{i=1}^k \sum_{v \in P_i^*} d(v,f_i^*) & \leq \sum_{i=1}^k \sum_{v \in P_i^*} d(v,f_v) + 2 d(v,c_i^*)\\
        &= \sum_{v \in V} d(v,f_v) + 2 \sum_{i=1}^k \sum_{v \in P_i^*} d(v,c_i^*)\\
        &= cost(F) + 2 \opt
    \end{align*}
\end{proof}

We can thus use a constant approximation for regular $k$-median to obtain a facility set of size $k$ and can then calculate the optimal solution of the resulting connected $k$-median with facilities instance to obtain a constant approximation for the problem in FPT time:

% \begin{theorem}\label{thm:tw_median}
%     Given a nice tree decomposition of a graph $G$ with tree-width $w-1$ and $O(nw)$ bag nodes and a natural number $k$, one can calculate a $(2\alpha +4)$-approximation for connected $k$-median in FPT-time with respect to the parameters $w$ and $k$, where $\alpha$ is the approximation ratio of a given $k$-median approximation algorithm.
% \end{theorem}

\begin{theorem}\label{thm:tw_median}
    Given a nice tree decomposition of a graph $G$ with tree-width $w$ and $O(nw)$ bag nodes and a natural number $k$, as well as an $\alpha$-approximation algorithm for unconnected $k$-median with running time $p(n)$, one can calculate a $(2\alpha +4)$-approximation for connected $k$-median in $O\left(p(n) + n \cdot (wk)^{O(w)}\right)$ time , which lies in FPT with respect to the parameter $\max(w,k)$.
\end{theorem}

%$p_1(n) + n \cdot (wk)^{O(w)}$

\begin{proof}
   We can use the approximation algorithm to calculate an $\alpha$-approximation for the non-connected $k$-median problem on $V$ with metric $d$ in $p(n)$. Let $F$ be the resulting center set. Given that removing the connectivity constraint only decreases the cost, we have $cost(F) \leq \alpha \opt$. By Lemma~\ref{lem:facility_nesting_median} the cost of the optimal connected $k$-median with facilities solution using $F$ as the facilities is upper bounded by $(\alpha + 2) \opt$. By Theorem~\ref{thm:con_median_facilities} we can calculate the respective clusters $P_1,...,P_k$ and center $c_1,...,c_k$ in $O(n \cdot (wk)^{O(w)})$. By choosing for every cluster $P_i$ the vertex $v \in P_i$ minimizing $d(v,f_i)$ as the center the distance of every vertex to the center is bounded by twice the distance to $f_i$ and we obtain a connected $k$-median solution with cost $(2\alpha + 4) \opt$.
\end{proof}

By combining this with the $(2 + \epsilon)$-approximation for $k$-median by Cohen-Addad et al.~\cite{cohen2approx} we obtain the following corollary:

\begin{corollary}
    Given a nice tree decomposition of a graph $G$ with tree-width $w$ and $O(nw)$ bag nodes and a natural number $k$. One can calculate a $(8 + \epsilon)$-approximation for connected $k$-median in FPT-time with respect to the parameter $\max(w,k)$.
\end{corollary}

One may note that the dynamic program for the connected $k$-median with facilities problem does also work if $d$ is not a metric. Accordingly we can also solve the connected $k$-means with facilities problem with it. Given that $k$-means fulfills $(8,2)$-nesting \cite{lin2010general}, we can obtain a similar result as Lemma~\ref{lem:facility_nesting_median}. Using this we end up with a similar approximation result as for $k$-median:

\begin{theorem}\label{thm:tw_means}
    Given a nice tree decomposition of a graph $G$ with tree-width $w$ and $O(nw)$ bag nodes and a natural number $k$, as well as an $\alpha$-approximation algorithm for unconnected $k$-means with running time $p(n)$, one can calculate a $(8\alpha +32)$-approximation\footnote{An additional factor of two is added because we cannot use the regular triangle inequality when replacing the facilities by vertices that are contained in the clusters to obtain feasible centers.} for connected $k$-means in $O\left(p(n) + n \cdot (wk)^{O(w)}\right)$ time, which is in FPT with respect to the parameter $\max(w,k)$.
\end{theorem}

\section{Hardness of approximation for connected \texorpdfstring{$k$}{k}-center}\label{sec:hardness}

%\Jan{Refer to $\phi$ where it makes sense, maybe find better name}

In this section we will show that it is NP-hard to obtain a constant approximation for the assignment version of the connected $k$-center problem. As we prove later, this also carries over to the non-assignment version. The hardness is based on a reduction from the 3-SAT problem. In this problem we are given a Boolean formula $\phi$ over $n$ variables $x_1,\ldots,x_n$ consisting of $m$ clauses $\cla_1,\ldots,\cla_m$, where each $\cla_j$ contains at most three literals from the set $\{x_i, \overline{x}_i \mid i \leq n\}$. The question is whether one can assign truth values to the variables such that each clause contains at least one literal that is true. Previous reductions already used that in connected $k$-center we can model the binary decision whether a variable (or a literal to be more precise) is true by assigning a corresponding point to one of two possible centers \cite{ge2008joint, drexler2024connected}. By doing this, one obtains that it is NP-hard to approximate the assignment version of connected $k$-center with an approximation factor better than $3$ even if $k = 2$. The main idea to obtain stronger inapproximability results is to increase the number of centers to create multiple layers, where in each layer the assignment of the vertices to the centers needs to increase the radius by an additive value if the formula cannot be fulfilled. More precisely the centers will only be connected to the first layer and each of them is paired up with another center such that they model the Boolean formula. The vertices corresponding to the literals of the formula connect them to an instance of the second layer. In this layer the centers are then again paired up (multiple times) with other centers to again model the formula where again the literals connect them with the next layer and so on. 

By using sufficiently many centers and a fitting metric, one can ensure that if the formula cannot be fulfilled there exists a vertex in the last layer such that on the path connecting it to its center we have an undesired assignment (corresponding to an unfulfilled clause or an invalid truth value of a Boolean variable) in each layer. This results in a radius linear in the number of layers while the minimum radius would have been $1$ if the formula could have been fulfilled.

Let us first discuss how we can model the formula $\phi$ using two centers. We create a connectivity graph $G = (V,E)$ with the vertex set

\begin{equation*}
    V = \{T,F\} \cup \{x_i,\overline{x}_i,a_i\mid i \in [n]\} \cup \{b_i\mid i \in [m]\},
\end{equation*}
where $T$ and $F$ are chosen as centers. We connect these vertices by adding the following edges to $E$:

\begin{itemize}
    \item For every $i \in [n]$, we connect $T$ and $F$ to $x_i$ and $\overline{x}_i$, respectively. Additionally we also connect $a_i$ to $x_i$ and $\overline{x}_i$.
    \item For every $j \in [m]$ and any literal $l$ contained in $\cla_j$, we connect the vertex corresponding to $l$ with the vertex $b_j$.
\end{itemize}

As for a distance metric we map our vertices to positions in $\mathbb{N} \cup \{\perp\}$ and use the following distance function:
\begin{equation*}
d'(a,b) = \begin{cases}
  0  & \text{if }a = b, \\
  1 & \text{if } a = \perp \oplus \text{ } b = \perp,\\
  2 & \text{else}.
\end{cases}
\end{equation*}
We place $T$ and the vertices $\{b_j\}_{j \in [m]}$ at position $1$, $F$ and the vertices in $\{a_i\}_{i \in [n]}$ at position $2$ and $\{x_i, \overline{x}_i\}_{i \in [n]}$ at position $\perp$. As a result the vertices $\{b_j\}_{j \in [m]}$ have distance $0$ to $T$ and $2$ to $F$ while for $\{a_i\}_{i \in [n]}$ it is the other way around and all vertices corresponding to literals have distance $1$ to both $T$ and $F$. An example of an instance created by this reduction can be found in Figure~\ref{fig:reduction_one_layer}. One might note that one could have created such distances using a simpler metric but the advantage of our definition is that it can be easily extended once we introduce additional layers. 

\begin{figure}
\centering
    \includegraphics[width= 0.5\textwidth]{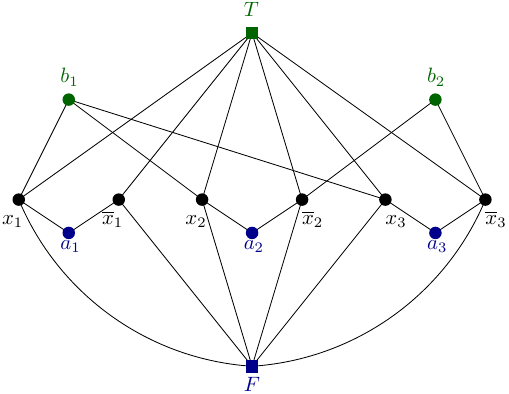}
    \caption{The connectivity graph corresponding to the 3-SAT formula $(x_1 \lor x_2 \lor x_3) \land (\overline{x}_2 \lor \overline{x}_3)$. Green vertices are placed at position $1$, blue ones at $2$ and black ones at $\perp$.}\label{fig:reduction_one_layer}
\end{figure}

Any assignment to the centers in this graph with radius $1$ can be transformed into a solution of the formula $\phi$ and vice versa. The idea behind this is that we can interpret all literals whose corresponding vertices are assigned to $T$ as being true and all assigned to $F$ to be false. Since for any $i \in [n]$ the distance of $a_i$ to $T$ is $2$, we know that $a_i$ needs to be assigned to $F$ which together with the connectivity constraint implies that $x_i$ or $\overline{x}_i$ is assigned to $F$. Intuitively this ensures that we do not allow $x_i$ to be true and false at the same time. For every $j \in [m]$ the vertex $b_j$ has distance $2$ to $F$ which means that $b_j$ needs to be assigned to $T$. Since $b_j$ is only neighbored to vertices that correspond to literals contained in $\cla_j$, we know that one of these vertices is assigned to $T$ which implies that at least one of the literals in $\cla_j$ is set to be true and $\cla_j$ is fulfilled. Since this holds for any clause we know that the resulting truth values fulfill the formula $\phi$ if we consider an assignment with radius $1$. The transformation of a solution of $\phi$ into an assignment works similarly. A more formal argument can be found in the paper by Drexler et al.~\cite{drexler2024connected}.

Before providing our detailed construction, we want to sketch how one can add an additional layer to obtain an instance with two layers. The addition of further layers follows the same ideas. If we have two layers the positions of the vertices consist of two entries and thus are elements of $(\mathbb{N} \cup \{\perp\})^2$. For an arbitrary $L \in \mathbb{N}$ the distance $d$ of two tuples $\alpha,\beta \in (\mathbb{N} \cup \{\perp\})^L$ is defined as:
\begin{equation*}
    d(\alpha,\beta) = \sum_{i=1}^{L} d'(\alpha_i,\beta_i).
\end{equation*}

In our construction we use multiple copies of the instance only using one layer. For a $p \in [2]$, $t,f \in [3]$ with $t \neq f$ let $I(p,\{t,f\})$ denote a copy of the one layer instance where the centers are replaced with normal vertices, the first coordinate of the position of every vertex is $p$ and the second coordinate is $t$ if the vertices are placed at the same position as $T$ and $f$ if they are placed at the same position as $F$ (otherwise the second coordinate stays $\perp$). The only vertices in this subgraph that are connected with the remainder of the instance will be $T$ and $F$. Let us assume that $T$ and $F$ are assigned to centers whose second coordinate is $t$ and $f$, respectively. By the same reasoning as for the instance with one layer, we obtain that an assignment adding the vertices $I(p,\{t,f\})$ with second coordinate $t$ to the same cluster as $T$ and the ones with second coordinate $f$ to the same cluster as $F$, corresponds to a solution of the formula $\phi$ and vice versa. 
%check42
%By the same reasoning as for the instance with one layer, we obtain that an assignment to these centers, such that any vertex in $I(p,\{t,f\})$, whose second coordinate is not $\perp$, gets assigned to the center with the same second coordinate, corresponds to a solution of the formula $\phi$ and vice versa. \Jan{fix sentence}

If we further assume that the first coordinate of the centers, to which $T$ and $F$ get assigned to, is $p$ and that there exists a solution for $\phi$, we know that it is possible to assign all vertices in $I(p,\{t,f\})$ such that the distance of any vertex to its center is at most $1$. If we assume instead that the first coordinate of the centers, to which $T$ and $F$ get assigned to, is unequal $p$ and that there exists no solution for $\phi$, we know that at least one vertex in $I(p,\{t,f\})$ gets assigned to a center whose position differs in both coordinates, which results in a distance of $4$. The goal of our construction is to ensure that there always exists a copy of an instance $I(p,\{t,f\})$ such $T$ and $F$ get assigned to centers whose coordinates are unequal to $p$ if $\phi$ cannot be fulfilled, while there exists an assignment where this does not happen if there is a solution for $\phi$.

For every $a \in [3]$ we have two centers $T_a$ and $F_a$ which are placed at $(1,a)$ and $(2,a)$. Additionally we have for every $i \in [n]$ two vertices $x_{i,a}$ and $\overline{x}_{i,a}$ that are both placed at position $(\perp,a)$ and connected to $T_a$ and $F_a$. Instead of adding vertices to represent the clauses and ensure that not both $x_{i,a}$ and $\overline{x}_{i,a}$ get assigned to $T_a$, we will use the copies of the instance with one layer.

For every $a,b \in [3]$ with $a < b$ and any pair of indices $i_a,i_b \in [n]$, we add a copy of $I(2,\{a,b\})$ to our instance where we connect the vertex $T$ of this instance with $x_{i_a,a}$ and $\overline{x}_{i_a,a}$ and the vertex $F$ with $x_{i_b,b}$ and $\overline{x}_{i_b,b}$. We say for any $e \in [3]$ and any $i \in [n]$ that $x_i$ is \emph{violated} regarding $e$ if both $x_{i,e}$ and $\overline{x}_{i,e}$ get assigned to $T_e$. We also say that $e$ itself is violated if there exists any $i \in [m]$ such that $x_i$ is violated regarding $e$. If both $a$ and $b$ are violated by an assignment then there exist indices $i_a$ and $i_b$ that are violated regarding $a$ and $b$ which means that $x_{i_a,a}$ and $\overline{x}_{i_a,a}$ are assigned to $T_a$ and $x_{i_b,b}$ and $\overline{x}_{i_b,b}$ are assigned to $T_b$. By our construction we know that there exists a copy of $I(2,\{a,b\})$ that is only neighbored to these vertices which implies that every vertex contained in it is assigned to a center whose first coordinate is $1 \neq 2$.

Similarly we say for any $e \in [3]$ and any $j \in m$ that $\cla_j$ is \emph{fulfilled} regarding $e$ if there exists a $x_i \in \cla_j$ such that $x_{i,e}$ is assigned to $T_e$ or if there exists a $\overline{x}_{i} \in \cla_j$ such that $\overline{x}_{i,e}$ is assigned to $T_e$. We call $e$ unfulfilled if there exists a clause $\cla_j$ that is not fulfilled regarding $e$. For every $a,b \in [3]$ with $a < b$ and any pair of indices $j_a,j_b \in [m]$ we add a copy of $I(1,\{a,b\})$ to our graph. We connect the vertex $T$ in this copy to the vertices $x_{i,a}$ and $\overline{x}_{i,a}$ corresponding to literals contained in $\cla_{j_a}$ and the vertex $F$ with the vertices $x_{i,b}$ and $\overline{x}_{i,b}$ corresponding to literals contained in $\cla_{j_b}$. As in the previous section it holds that if $a$ and $b$ are unfulfilled there exists a copy of $I(1,\{a,b\})$ such that every vertex in that copy is assigned to a center whose first coordinate is $2 \neq 1$. A sketch depicting the resulting structure of the connectivity graph can be found in Figure~\ref{fig:reduction_two_layers}.

\begin{figure}
\centering
    \includegraphics[width= 0.9\textwidth]{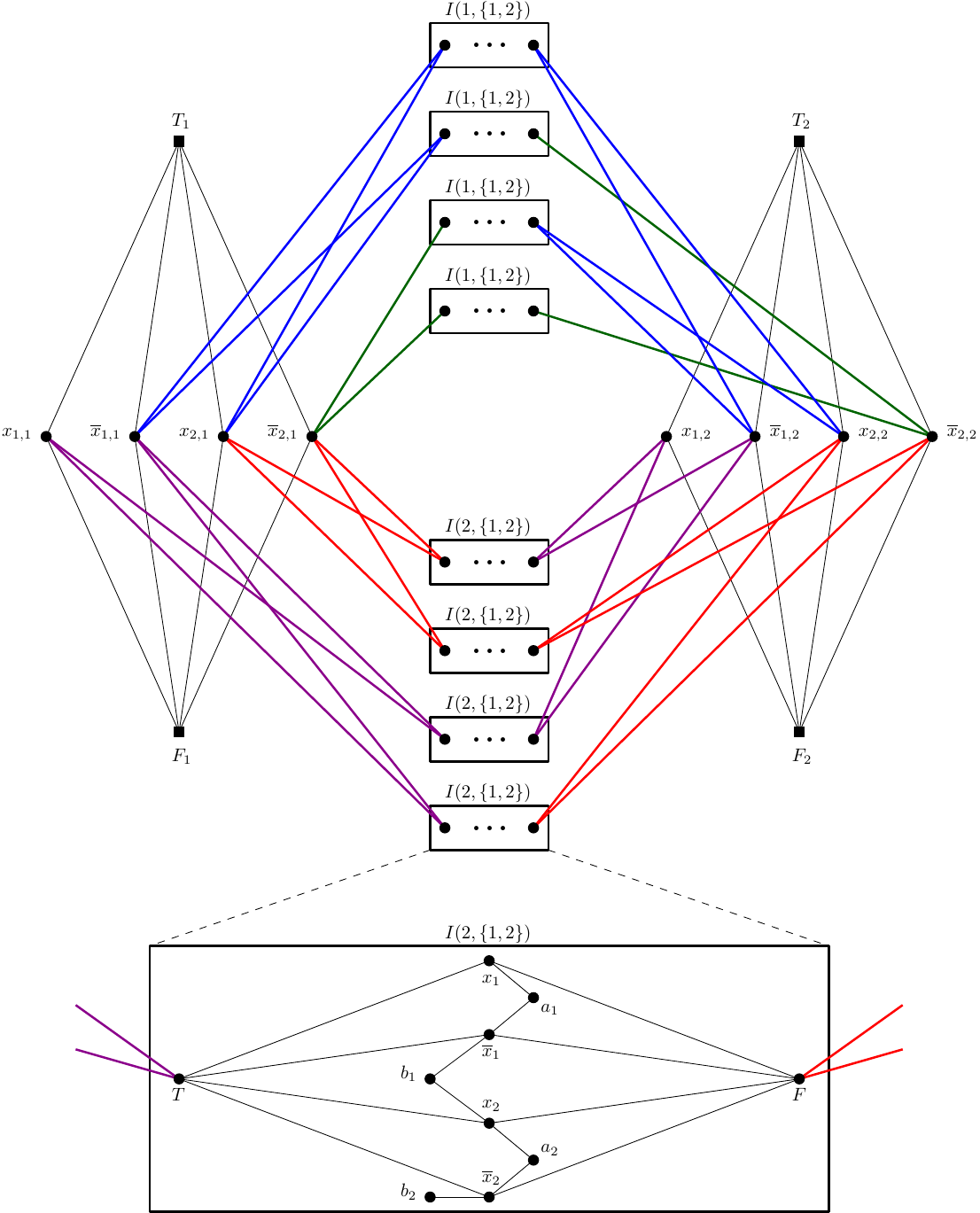}
    \caption{A subgraph of the connectivity graph with two layers corresponding to the 3-SAT formula $ (\overline{x}_1 \lor x_2) \land (\overline{x}_2)$. Blue edges correspond to the clause $(\overline{x}_1 \lor x_2)$, green ones to the clause $(\overline{x}_2)$, magenta ones to the variable $x_1$ and red ones to the variable $x_2$. The vertices and edges in the subgraph $I(p,\{1,2\})$ (with $p \in [2]$) are for the most part not depicted. The complete connectivity graph also contains the vertices $T_3$, $F_3$, $x_{1,3}$, $\overline{x}_{1,3}$, $x_{2,3}$ and $\overline{x}_{2,3}$, as well as $8$ copies of one-layer instances connecting them with the vertices with first coordinate $1$ and $8$ copies connecting them with the vertices with first coordinate $2$. }\label{fig:reduction_two_layers}
\end{figure}

If there exists a solution for $\phi$ we can assign for each $e \in [3]$ the vertices $\{x_{i,e},\overline{x}_{i,e}\}_{i \in [n]}$ according to this solution and do the same within the copies of the 1-layer instance to obtain a radius of $1$. If there is no solution for $\phi$ it holds for every assignment of the vertices that every $e \in [3]$ is violated or unfulfilled. Thus there exist two unfulfilled or two violated indices $a,b \in [3]$ with $a \neq b$. This means that there exists a $p \in [2]$ and a copy of $I(p,\{a,b\})$ such that every vertex in that copy is assigned to centers whose first coordinate is unequal $p$. By our previous argument, this results in a radius of $4$. This implies that it is NP-hard to approximate the assignment version of connected $k$-center with an approximation factor smaller than $4$.

Now we provide a formal construction for an arbitrarily chosen number $L$ of layers. This construction creates an instance for the connected clustering problem that provides an assignment with radius $1$ if the formula $\phi$ can be fulfilled and cannot be solved with a radius smaller $2 L$ otherwise. 

Let $\mathcal{S}_1,\ldots,\mathcal{S}_L$ be natural numbers that will be defined later fulfilling that $\mathcal{S}_1 = 2$ and $\mathcal{S}_i < \mathcal{S}_j$ for any $ i < j$. We define a distance metric $(\mathcal{V}_L,d)$, with set $\mathcal{V}_L = ([\mathcal{S}_1]\cup\{\perp\}) \times ([\mathcal{S}_2]\cup\{\perp\})\times \ldots \times ([\mathcal{S}_L]\cup\{\perp\})$ and metric $d$: For each $\alpha = (\alpha_1,\ldots,\alpha_L), \beta = (\beta_1,\ldots,\beta_L) \in \mathcal{V}_L$ it is defined as $d(\alpha,\beta) = \sum_{i=1}^L  d'(\alpha_i,\beta_i)$ where for all $a,b\in \mathbb{N} \cup \{\perp\}$:
\begin{equation*}
d'(a,b) = \begin{cases}
  0  & \text{if }a = b, \\
  1 & \text{if } a = \perp \oplus \text{ } b = \perp,\\
  2 & \text{else.}
\end{cases}
\end{equation*}

We iteratively define an instance with $k = \prod_{i = 1}^L \mathcal{S}_i$ many centers and $L$ layers. For a fixed $l \in \{0,\ldots, L\}$ and a tuple $\pi^l = (\pi^l_1,\ldots,\pi^l_l) \in [\mathcal{S}_1] \times [\mathcal{S}_2] \times \ldots \times [\mathcal{S}_l]$ and a tuple of sets $\Psi^l =(\Psi^l_1,\ldots,\Psi^l_{L-l})$ where for all $i \leq L - l$ it holds that $\Psi^l_i \subseteq [\mathcal{S}_{i + l}]$ and $|\Psi^l_i| =  \mathcal{S}_{i}$ we define the instance $I(l,\pi^l, \Psi^l)$ as follows:

For $l = L$ it holds that $I(L,\pi^L)$ only contains a single point $Q_L$ that is placed at position $\pi^L$. Let now for a fixed $l + 1 > 0$ and all fitting tuples $\pi^{l+1} \in [\mathcal{S}_1] \times \ldots \times [\mathcal{S}_{l+1}]$ and $\Psi^{l+1} \in \mathbb{P}([\mathcal{S}_{l+2}])\times \ldots \times \mathbb{P}([\mathcal{S}_{L}])$ ($\mathbb{P}(S)$ denotes the power set of a set $S$) the instance $I(l + 1,p^{l+1},\Psi^{l+1})$ be already defined. We will use these instances in the definition of $I(l,\pi^l,\Psi^l$).

One may note that $\Psi^l_{1}$ contains only two numbers. Let us denote the smaller one as $t$ and the larger one as $f$.

For every $\epsilon \in \Psi^l_{2} \times \ldots \times \Psi^l_{L-l}$, we add a vertex $Q_{l, t \cdot \epsilon}$ which is also called $T_{l,\epsilon}$ that is placed at position $\pi^l \cdot t \cdot \epsilon$ to $I(l,\pi^l,\Psi^l)$, where we define $\cdot$ as the concatenation of two tuples. Similarly we also add a vertex $Q_{l, f \cdot \epsilon}$ at position $\pi^l \cdot f \cdot \epsilon$ which is also called $F_{l,\epsilon}$. Additionally for every $i \in \{1,\ldots,n\}$ we will add two vertices $x_{i,\epsilon}$ and $\overline{x}_{i,\epsilon}$. Both vertices are placed at position $\pi^l \cdot \perp \cdot \epsilon$ and both are connected to $T_{l, \epsilon}$ and $F_{l, \epsilon}$ via an edge, respectively.

Furthermore for each $\Psi^{l+1}$ of length $L - l -1$ with $\Psi^{l+1}_i \subset \Psi^l_{i + 1}$ and $|\Psi^{l+1}_i| = \mathcal{S}_{i}$ there exist multiple copies of the instance $I(l + 1, \pi^{l}\cdot t, \Psi^{l+1})$ and $I(l + 1, \pi^{l}\cdot f, \Psi^{l+1})$. To be more precise let  $ \{\epsilon^{(1)},\ldots,\epsilon^{(u)}\} = 
\Psi^{l+1}_1 \times\ldots \times \Psi^{l+1}_{L-l-1}$, where the $\epsilon^{(i)}$ are in lexicographic order. We create a copy of $I(l + 1, \pi^l \cdot t, \Psi^{l+1})$ for each tuple $(j_1,\ldots,j_u) \in [m]^u$ which we denote as $I(l + 1, \pi^l \cdot t, \Psi^{l+1})[j_1,\ldots,j_u]$. For every $u' \in [u]$ we connect the vertex $Q_{\epsilon^{(u')}}$ in the respective copy with all $x_{i,\epsilon^{(u')}}$ such that $x_i \in \cla_{j_{u'}}$ and all $\overline{x}_{i,\epsilon^{(u')}}$ such that $\overline{x}_i \in \cla_{j_{u'}}$. In the same fashion we create a copy of $I(l + 1, \pi^l \cdot f, \Psi^{l+1})$ denoted as $I(l + 1, \pi^l \cdot f, \Psi^{l+1})[i_1,\ldots,i_u]$ for every tuple $(i_1,\ldots,i_u) \in [n]^u$ and for every $u' \in [u]$ we connect the vertex $Q_{\epsilon^{u'}}$ in this instance with $x_{i_{u'},\epsilon^{(u')}}$ and $\overline{x}_{i_{u'},\epsilon^{(u')}}$.

By applying this construction inductively, one ends up with the instance $I(0,([\mathcal{S}_1],\ldots,[\mathcal{S}_L]))$. In this instance we choose for every $\epsilon \in [\mathcal{S}_1] \times \ldots \times [\mathcal{S}_L]$ the vertex $Q_{0,\epsilon}$ as a center. 

To prove the correctness of this reduction we have to fix the values $\mathcal{S}_1,\ldots.,\mathcal{S}_L$. As mentioned above $\mathcal{S}_1 = 2$ and for every $l \in \mathbb{N}$ with $l \geq 2$ we define:
\begin{equation*}
\mathcal{S}_{l+1} =  (\mathcal{S}_l -1) 2\prod_{i=1}^{l-1} \binom{\mathcal{S}_{i+1}}{\mathcal{S}_i} + 1
\end{equation*}

%This values are chosen in such a way that the following technical property holds:

\begin{figure}
    \centering
    \includegraphics[width = 0.3\textwidth]{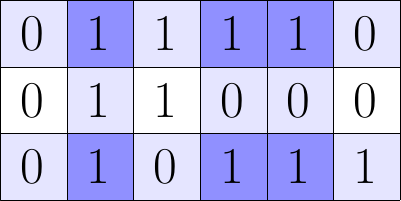}
    \caption{A binary $3\times6$ tensor with a constant $2 \times 3$ subtensor (blue).}\label{fig:subtensor}
\end{figure}

These values are chosen in such a way that they fulfill a certain technical property which can be expressed using \emph{tensors}:

\begin{definition}
    For a given natural number $o$ and $d_1,...,d_o$ we say that a binary $o$-th-order $d_1 \times d_2 \times \ldots \times d_o$-tensor is a function $T:[d_1] \times [d_2] \times \ldots \times [d_o] \rightarrow \{0,1\}$. 
    For any $S_1 \subseteq [d_1]$, $S_2 \subseteq [d_2]$,\ldots, $S_o \subseteq [d_o]$ we say that the entries for $S_1 \times \ldots \times S_o$ form a $|S_1| \times  \ldots \times |S_o|$ subtensor of $T$. We say that this subtensor is constant if either all of its entries are equal to zero or all of its entries are equal to $1$.
\end{definition}

The easiest way to think about binary tensors is probably just to imagine them as a $d$-dimensional table filled with values in $\{0,1\}$. Figure \ref{fig:subtensor} depicts a two dimensional tensor (i.e.\ a matrix) and a constant subtensor this way. We can show that we have chosen the $\mathcal{S}_i$ values in such a way that any binary subtensor using them as dimensions contains a constant subtensor of sufficient size:

\begin{lemma}
\label{lem:tensors}
For any $l \leq L$, $l \geq 2$ it holds that any $(l-1)$th-order $\mathcal{S}_2 \times \ldots \times \mathcal{S}_l$-tensor filled with binary values contains a constant $\mathcal{S}_1 \times\ldots \times \mathcal{S}_{l-1}$-subtensor. 
\end{lemma}

\begin{proof}
A simple calculation gives that $\mathcal{S}_2 = 3$. Since any binary $3\times 1$ tensor contains either two times $0$ or two times $1$, the statement trivially holds for $l = 2$.

Let us now assume the statement holds for a fixed $l$. We will consider an arbitrary $\mathcal{S}_2 \times\ldots. \times \mathcal{S}_{l+1}$ tensor. By the assumption it holds that for any choice of the last index of the tensor, the resulting $\mathcal{S}_2 \times\ldots \times \mathcal{S}_{l}\times 1$ tensor contains a constant $\mathcal{S}_1 \times\ldots \times \mathcal{S}_{l-1}\times 1$ subtensor. If we can find $\mathcal{S}_l$ choices of the last index such that the resulting subtensors are the same, we have proven the statement for $l+1$. To do this, we will calculate the number of possible subtensors, this is exactly:
\begin{equation*}
2 \prod_{i= 1}^{l - 1} \binom{\mathcal{S}_{i+1}}{\mathcal{S}_i}
\end{equation*}
as there are two choices for the value contained in the subtensor and $\binom{\mathcal{S}_{i+1}}{\mathcal{S}_i}$ possibilities to choose the $\mathcal{S}_i$ indices in dimension $i$ among the $\mathcal{S}_{i+1}$ choices. Thus we can apply the pigeonhole principle to conclude that at least one subtensor appears $\mathcal{S}_l$ times. This directly implies that there exists a constant $\mathcal{S}_1 \times\ldots \times \mathcal{S}_{l}$ subtensor which directly proves the lemma.
\end{proof}

This property ensures that if $\phi$ cannot be fulfilled then for every $l\in [L-1]$ every instance $I(l,\pi^l,\Psi^l)$ contains a copy of an instance $I(l + 1,\pi^{l+1},\Psi^{l+1})$ where all vertices are assigned to a center whose $(l+1)$-th coordinate is unequal $\pi_{l+1}^{l+1}$. By inductively going over all layers from $0$ to $L$, we can use this to identify a vertex that is assigned to a center whose position differs in every coordinate, which results in a radius of $2L$:

\begin{lemma}
\label{lem:red_center_lower}
If there does not exist a solution for $\phi$, every disjoint assignment of the vertices to the centers results in a radius of at least $2L$.
\end{lemma}

\begin{proof}
 We will show inductively that for every $l \in \{0,\ldots,L\}$ there exist a $\pi^l  \in [\mathcal{S}_1] \times \ldots \times [\mathcal{S}_L]$ and a $\Psi^l = (\Psi^l_1,\ldots,\Psi^l_{L-l}) \in \mathbb{P}([\mathcal{S}_{l+1}])\times \ldots \times \mathbb{P}([\mathcal{S}_L])$ and a copy of instance $I(l,\pi^l,\Psi^l)$ such that every vertex in it gets assigned to a center $Q_{0,\epsilon}$ such that for ever $j \leq l$ it holds that $\epsilon_j \neq \pi^l_j$. One might observe that for $l = 0$ this claim trivially holds.
 
Let us assume that the claim holds for a fixed $l$. We will now consider the situation for $l+1$. By our assumption, there exists an instance $I(l,\pi^l,\Psi^l)$ such that every vertex in this instance gets assigned to a center whose first $l$ coordinates are unequal to $\pi^l$. Let $\epsilon \in \Psi^l_{2} \times \Psi^l_{3} \times \ldots \times \Psi^l_{L-l}$. For an arbitrary $i \leq n$, we say that $x_i$ is true regarding $\epsilon$ if $x_{i,\epsilon}$ gets assigned to a center whose $(l+1)$-th coordinate is equal to $t$ and we say that $x_i$ is false if $\overline{x}_{i,\epsilon}$ is assigned to a center whose $(l+1)$-th coordinate is $t$. If $x_i$ is both true and false regarding $\epsilon$, we say that $x_i$ is violated. Now there are two cases that can occur: If no variable is violated regarding $\epsilon$ then the assignment of the $x_{i,\epsilon}$ corresponds to an assignment of truth values to the Boolean variables $x_i$. Since our original formula cannot be satisfied, there has to exist a $j \in [m]$ such that $\cla_j$ is not fulfilled by this assignment. In this case we say that $\epsilon$ is unfulfilled. In the other case there is at least one $i$ such that $x_i$ gets violated regarding $\epsilon$ and in this case we say that also $\epsilon$ gets violated.

Since this holds for all $\epsilon \in \Psi^l_{2} \times \Psi^l_{3} \times \ldots \times \Psi^l_{L - l}$ which are $\prod_{i=2}^{L-l} \mathcal{S}_i$ many we may use Lemma~\ref{lem:tensors} to conclude that there exists an $\Psi^{l+1}$ such that for all $i \in [L-l-1]$ it holds that $\Psi^{l+1}_i \subseteq \Psi^l_{i+1}$ and $|\Psi^{l+1}_i| = \mathcal{S}_{i}$ and all $\epsilon \in \Psi^{l+1}_{1} \times \Psi^{l+1}_2 \times \ldots \times \Psi^{l+1}_{L-l-1}$ are unfulfilled or all of them get violated.

Let us first consider the case that they are unfulfilled. Then for each $\epsilon^{(h)}$ (which denotes the lexicographically $h$-th element in $\Psi^{l+1}_1 \times \Psi^{l+1}_2 \times \ldots \times \Psi^{l+1}_{L-l-1}$) there exists an $j_h \in [m]$ such that for every $x_i \in \cla_{j_h}$ and $\overline{x}_{i'} \in \cla_{j_h}$ the vertices $x_{i,\epsilon^{(h)}}$ and $\overline{x}_{i',\epsilon^{(h)}}$, respectively, are assigned to a center whose $i$-th coordinate is unequal $t$. Let $u = \prod_{i=1}^{L-l-2} \mathcal{S}_i$. Then it holds for the subinstance $I(l+1,\pi^l \cdot t,\Psi^{l+1})[j_1,\ldots,j_u]$ that for each $h \in [u]$ every neighbor of $Q_{l+1,\epsilon^{(h)}}$ outside $I(l+1,\pi^l \cdot t,\Psi^{l+1})[j_1,\ldots,j_u]$ gets assigned to a center whose $(l+1)$-th coordinate is unequal $t$. By combining this with the induction hypothesis we get that every neighbor of a vertex in the instance $I(l+1,\pi^l \cdot t,\Psi^{l+1})[j_1,\ldots,j_u]$ gets assigned to a center whose first $l+1$ coordinates are all unequal $\pi^l \cdot t$ which also means that the vertices in the instance get assigned to these centers which proves our claim for $l+1$.

If they are all violated, it holds similarly for each $\epsilon^{(h)}$ that there exists an $i_h \in [n]$ such that $x_{i_h}$ and $\overline{x}_{i_h}$ get assigned to centers whose $(l+1)$-th coordinates are unequal $f$. Thus all neighbors of the vertices in the instance $I(l+1,\pi^l \cdot f,\Psi^{l+1})[i_1,\ldots,i_u]$ are assigned to centers whose first $l+1$ coordinates are all unequal $\pi^l \cdot f$ which again proves the claim. Thus the induction holds.

Since the claim holds for any $l \leq L$ we may in particular conclude that there exists an $\pi^L$ and a copy of  $I(L,\pi^L)$ such that the vertex $Q_L$ in this instance gets assigned to a center $Q_{0,\epsilon}$ such that for every $l \leq L$ it holds that $\pi^L_l \neq \epsilon_l$. Thus by the definition of the distance function the radius of the respective cluster is at least $2L$ which proves the lemma.
\end{proof}

At the same time the instance is constructed in such a way that if the formula $\phi$ can be fulfilled there exists an assignment of radius $1$:

\begin{lemma}
\label{lem:red_center_upper}
If there exist a solution fulfilling $\phi$ then there exists an assignment to the center vertices such that every vertex is only at distance $1$ from its respective center.
\end{lemma}

\begin{proof}
We will show the following claim for all $l \leq L$: For all instances $I(l,\pi^l, \Psi^l)$ it holds that if there exists an assignment such that for every $\delta \in \Psi^l_1 \times \Psi^l_2 \times \ldots \Psi^l_{L-l}$ the vertex $Q_{l,\delta}$ of the instance gets assigned to $Q_{0,\pi^l \cdot \delta}$ then there exists a assignment such that every vertex in $I(l,\pi^l,\Psi^l)$ has only distance $1$ to the center it gets assigned to.

Obviously the claim holds for $l = L$. Now we assume that for a fixed $l+1$ the claim is already true and consider an instance $I(l,\pi^l, \Psi^l)$ where each $Q_{l,\delta}$ gets assigned to $Q_{0,\pi^l\cdot \delta}$. Let $x_1^*,\ldots,x_n^*$ be an satisfying assignment for $\phi$. For all $\epsilon \in \Psi^l_2 \times \Psi^l_3 \times \ldots \times \Psi^l_{L-l}$ we assign the vertex $x_{i,\epsilon}$ to $Q_{0,\pi^l \cdot t \cdot \epsilon}$ and $\overline{x}_{i,\epsilon}$ to $Q_{0,\pi^l \cdot f \cdot \epsilon}$ if $x_i^* = true$. Otherwise we assign $x_{i,\epsilon}$ to $Q_{0,\pi^l \cdot f \cdot \epsilon}$ and $\overline{x}_{i,\epsilon}$ to $Q_{0,\pi^l \cdot t \cdot \epsilon}$. One might observe that this assignment does not violate the connectivity since $T_{l,\epsilon}$ and $F_{l,\epsilon}$ have been assigned to $Q_{0,\pi^l\cdot t\cdot \epsilon}$ and $Q_{0,\pi^l\cdot f\cdot \epsilon}$, respectively. Also in both cases the points are only at distance $1$ from the respective center.

Now we consider an arbitrary $\Psi^{l+1}$ fulfilling for all $i \in [L-l-1]$ that $\Psi^{l+1}_i \subseteq \Psi^l_{i+1}$ and $|\Psi^{l+1}_i| = \mathcal{S}_i$. Let $u = \prod_{i = 1}^{L-l - 1} \mathcal{S}_i$ and let $\{\epsilon^{(1)},\ldots,\epsilon^{(u)}\} = \Psi^{l+1}_1 \times \Psi^{l+1}_2 \times \ldots \times \Psi^{l+1}_{L-l-1}$. For each $(i_1,\ldots,i_u) \in [n]^u$ we consider the instance $I(l+1,\pi^l \cdot f,\Psi^{l+1})[i_1,\ldots,i_u]$. For each $u' \in [u]$ we assign the vertex $Q_{l+1,\epsilon^{(u')}}$ in this instance to the center $Q_{0,\pi^l\cdot f \cdot \epsilon^{(u') }}$. One might note that the distance to this center is $0$. Furthermore this assignment does not violate the connectivity since either $x_{i_u, \epsilon^{(u)}}$ or $\overline{x}_{i_u, \epsilon^{(u)}}$ is assigned to $Q_{0,\pi^l\cdot f \cdot \epsilon^{(u') }}$ and both vertices are connected to $Q_{l+1,\epsilon^{u'}}$. Now one might observe that in instance $I(l+1,\pi^l \cdot f,\Psi^{l+1})[i_1,\ldots,i_u]$ for each $\delta \in \Psi^{l+1}_1 \times \Psi^{l+1}_2 \times \ldots \times \Psi^{l+1}_{L-l-1}$ the vertex $Q_{l,\delta}$ gets assigned to $Q_{0,\pi\cdot f \cdot \delta}$. Thus we can apply our assumption to conclude that there exists an assignment of all the vertices in this instance with radius $1$. 

Similarly for any $(j_1,\ldots,j_u) \in [m]^u$ we consider the instance $I(l+1,\pi^l \cdot t,\Psi^{l+1})[j_1,\ldots,j_u]$. For each $u' \in [u]$ we assign the vertex $Q_{l+1,\epsilon^{(u')}}$ to the center $Q_{0,\pi^l\cdot t\cdot \epsilon^{(u')}}$ which is at distance $0$. To show that this assignment does not violate the connectivity we consider the clause $\cla_{j_{u'}}$. Obviously this clause is fulfilled by $x^*$. There exist two cases. If there exist an $i $ with $x_i^* = true$ and $x_i \in \cla_j$ then $x_{i,\epsilon^{(u')}}$ has been assigned to $Q_{0,\pi^l\cdot t\cdot \epsilon^{(u')}}$ while also being a neighbor of $Q_{l+1,\epsilon^{(u')}}$. Otherwise there exists an $i$ with $ \overline{x}_i \in \cla_{j^{(u')}}$ and $x_i^* = false$ which means that $\overline{x}_{i,\epsilon^{(u')}}$ has been assigned to $Q_{0,\pi^l\cdot t\cdot \epsilon^{(u')}}$ and is a neighbor of $Q_{l+1,\epsilon^{(u')}}$. In both cases the connectivity is not violated. Since this assignment can be done for all $\epsilon^{(u')}$, we can again apply our induction hypothesis to conclude that all vertices in $I(l+1,\pi^l \cdot t,\Psi^{l+1})[j_1,\ldots,j_u]$ can be assigned with radius $1$.

One might observe that we have assigned every vertex in $I(l,\pi^l,\Psi^l)$ to a center with radius $1$. Thus the claim also holds for $l$ itself and by induction we may conclude that it holds for all $l' \leq L$ including $L$ itself. Since for every $\delta \in \Psi^0_1 \times \ldots \times \Psi^0_L$ the center $Q_{0,\delta}$ can trivially be assigned to itself, we obtain that there exists an assignment of radius $1$ for the entire instance and the lemma is proven.
\end{proof}

To obtain hardness results for the assignment version we have to show that for a fixed number $L$ of layers the size of the instance we create is bounded polynomially in the size of $\phi$ and can be calculated in polynomial time. One may note that the respective polynomial depends on $L$. For simplicity reasons we denote for any $h \in \mathbb{N}$ the product $ \prod_{i = 1}^h \mathcal{S}_i$ by $\mathcal{P}_h$. Then $\mathcal{P}_L$ is exactly the number of centers of the instance with $L$ layers and we can bound $\mathcal{P}_h$ by $\mathcal{S}_{h+1}$:

\begin{equation}\label{eq:prod_layers}
        \mathcal{P}_h = \prod_{i = 1}^h \mathcal{S}_i = 2 \prod_{i = 1}^{h-1} \mathcal{S}_{i+1} \leq 2 \prod_{i=1}^{h-1} \binom{\mathcal{S}_{i+1}}{\mathcal{S}_i} \leq \mathcal{S}_{h+1}
\end{equation}

Now we bound the size of the constructed instance:

\begin{lemma}
\label{lem:red_center_polynomial}
For any constant $L$, the size of the instance $I(0,([\mathcal{S}_1],\ldots,[\mathcal{S}_L])$ as well as the time needed to calculate the instance is polynomial in $n$ and $m$.
\end{lemma}

\begin{proof}
For any instance $I(l,\pi^l,\Psi^l)$ we will call the number $L - l$ the height of the instance. Note that for any $l \in [L]$ all instances of height $h$ have the same number of vertices. We denote this number as $\mathcal{N}_h$. Our goal is now to bound $\mathcal{N}_L$. To do this, we will calculate the respective numbers recursively. Obviously $\mathcal{N}_0 = 1$.

Now we consider an arbitrary height $h \in [L]$ and an arbitrary instance $I(l, \pi^{l}, \Psi^{l})$ of that height, where $ l := L - h$. For every $\epsilon \in \Psi^{l}_2\times \ldots\times \Psi^{L-h}_l$ the instance contains the vertices $T_{l,\epsilon}$, $F_{l,\epsilon}$ as well as $x_{i,\epsilon}$ and $\overline{x}_{i,\epsilon}$ for any $i \in [n]$. Since there are $\prod_{i = 2}^h \mathcal{S}_i$ choices for $\epsilon$, this amounts in a total of $ 2(n + 1) \prod_{i = 2}^h \mathcal{S}_i = (n+1)\mathcal{P}_h$ vertices, which by Inequality~\eqref{eq:prod_layers} is upper bounded by $(n+1)\mathcal{S}_{h+1}$. Additionally for every $\Psi^{l+1}$ with $\Psi^{l+1}_i \subseteq \Psi^l_{i+1}$ and $|\Psi^{l+1}_i| = \mathcal{S}_i$ for any $i \in [L-l-1]$ and any tuple $(i_1,\ldots,i_u) \in [n]^u$ or $(j_1,\ldots,j_u) \in [m]^u$, a copy of an instance of height $h-1$ gets added, where $u = \prod_{g = 1}^{L - l - 1} | \Psi^{l+1}_g| = \mathcal{P}_{h - 1} \leq \mathcal{S}_h$. Since there are $\binom{\mathcal{S}_{g+1}}{\mathcal{S}_{g}}$ choices for $\Psi^{l+1}_g$ for any $g \in [h - 1]$ we get in total
\begin{align*}
\mathcal{N}_h &\leq (n + 1) \mathcal{S}_{h + 1}  + \prod_{i = 1}^{h - 1} \binom{\mathcal{S}_{i+1}}{\mathcal{S}_i} \left(n^{\mathcal{S}_{h}} + m^{\mathcal{S}_{h}} \right) \mathcal{N}_{h-1}\\
&\leq (n+1) \mathcal{S}_{h+1} + \mathcal{S}_{h+1} (n+ m)^{\mathcal{S}_h} \mathcal{N}_{h-1}\\
&\leq 2 \mathcal{S}_{h+1} (n+ m)^{\mathcal{S}_h} \mathcal{N}_{h-1}.
\end{align*}

Thus by solving the recursion we can bound $\mathcal{N}_L$ as follows:
\begin{align*}
\mathcal{N}_L &\leq \prod_{h=1}^L 2 \mathcal{S}_{h+1} (n+ m)^{\mathcal{S}_h}\\
&\leq 2^L \mathcal{P}_{L+1} (n+m)^{L \mathcal{S}_L}.
\end{align*}

Since $\mathcal{P}_{L+1}$ and $\mathcal{S}_L$ only depend on $L$, this is polynomial in $n + m$. The number of edges of the connectivity graph is then obviously upper bounded by $(\mathcal{N}_L)^2$ and it is easy to verify that every single vertex and edge can be added in polynomial time. Thus the lemma holds.
\end{proof}

By combining Lemma~\ref{lem:red_center_lower},~\ref{lem:red_center_upper} and~\ref{lem:red_center_polynomial} we directly obtain that it is NP hard to obtain an $\alpha$-approximation for the assignment version for a constant $\alpha \in \mathbb{N}$, as one could choose $L \geq \frac{\alpha}{2}$ to solve 3-SAT. To get a stronger inapproximability result, we show that $L \in \Omega(\log^*(\mathcal{P}_L))$. As the instance with $L$ layers has exactly $\mathcal{P}_L$ centers, this implies that it is NP-hard to approximate the assignment version of connected $k$-center with an approximation factor better than $\Theta(\log^*(k))$:

%\begin{theorem}
%Unless P = NP there exists no constant approximation algorithm for the assignment version of the connected $k$-center problem. 
%\end{theorem}

\begin{lemma}\label{lem:reduction_bound_prod}
    It holds that $L \in \Omega(\log^*(\mathcal{P}_L))$.
\end{lemma}

\begin{proof}

    By Equation~\eqref{eq:prod_layers} we know that $\mathcal{P}_L \leq \mathcal{S}_{L+1}$. We show via induction that $\log^*(\mathcal{S}_L) \leq (2L -2)$ for $L \geq 2$. For $L = 2$ we have that $\mathcal{S}_L = 3$ and $\log^*(3) = 2$ thus the statement is fulfilled.

    Let now the statement be true for a fixed $L$. We will first bound the value of $\mathcal{S}_{L+1}$ in terms of the value $\mathcal{S}_L$:
    \begin{align*}
        \mathcal{S}_{L+1} &= (\mathcal{S}_L -1) 2\prod_{i=1}^{L-1} \binom{\mathcal{S}_{i+1}}{\mathcal{S}_i} + 1\\
        &\leq 2 \mathcal{S}_L \prod_{i=1}^{L - 1} \frac{\mathcal{S}_{i+1}!}{\mathcal{S}_{i}!}\\
        &= 2 \mathcal{S}_L \frac{\mathcal{S}_{L}!}{\mathcal{S}_1!}\\
        &= \mathcal{S}_L \mathcal{S}_{L}!\\
        &\leq \left(\mathcal{S}_L\right)^{\mathcal{S}_L}.
    \end{align*}

    Then it holds that:
    \begin{equation*}
        \log(\log(\mathcal{S}_{L+1})) \leq \log(\mathcal{S}_L \log(\mathcal{S}_L)) = \log(\mathcal{S}_L) + \log(\log(\mathcal{S}_L)) \leq \mathcal{S}_L.
    \end{equation*}
    And thus:
    \begin{equation*}
        \log^*(\mathcal{S}_{L+1}) \leq \log^*(\mathcal{S}_L) + 2 \leq 2L.
    \end{equation*}

    Thus the induction holds and we can bound the value of $\log^*(\mathcal{P}_L)$ by $2L$:
    \begin{equation*}
        \log^*(\mathcal{P}_L) \leq \log^*(\mathcal{S}_{L+1}) \leq 2L.
    \end{equation*}
\end{proof}

So far all hardness results presented in this work only apply to the assignment version of connected $k$-center. However, while being able to choose the centers freely could possibly allow for different approximation approaches, one would not expect the problem to get much easier if also the centers have to be determined. Indeed one can show that any approximation algorithm for the non-assignment version can be used to approximate the assignment version while only losing a multiplicative factor of two. Thus the hardness results carry over:

\begin{theorem}\label{thm:reduction_centers_assignment}
    If there exists an $\alpha$-approximation algorithm for connected $k$-center with polynomial running time then there also exists a $2 \alpha$-approximation algorithm for the assignment version.
\end{theorem}

\begin{proof}

\begin{figure}
\begin{subfigure}{0.26\textwidth}
    \includegraphics[trim={3.73cm 0cm 3.73cm 0cm}, clip, width=\textwidth, page =1]{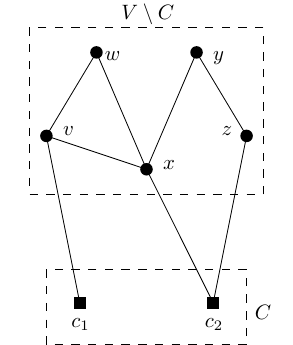}
    \caption{The graph $G$.}
\end{subfigure}
\hfill
\begin{subfigure}{0.65\textwidth}
    \includegraphics[width=\textwidth, page =2]{figures/reduction_remove_centers.pdf}
    \caption{The graph $G'$.}
\end{subfigure}
\caption{An example how the connectivity graph $G$ (left side) of an instance of the assignment version of connected $k$-center with $C = \{c_1,c_2\}$ can be transformed into a connectivity graph $G'$ where the centers are not fixed.}
\label{fig:remove_centers}
\end{figure}

    Let an $\alpha$-approximation for the non-assignment version of connected $k$-center be given. Let $G = (V,E)$ be a graph, $d: V \rightarrow \mathbb{R}_{\geq 0}$ a distance metric and $C = \{c_1,\ldots,c_k\} \subseteq V$ a given center set. Without loss of generality we may assume that there exist no edges between pairs of centers, as the centers need to be assigned to themselves anyway. To calculate a feasible assignment to $C$, we can transform this instance into an instance of the non-assignment version as follows:
    \begin{itemize}
        \item For any $v \in V \setminus C$ we add $k+1$ copies to our vertex set while for every $c \in C$ we only have a single copy. We end up with:
        
        $V' = \{v^j\mid v \in V \setminus C\land j \in [k+1]\} \cup C$.
        \item For all $v,w \in V \setminus C$ and $j \in [k+1]$ we add an edge $\{v^j,w^j\}$ to our new edge set $E'$ if $\{v,w\} \in E$.
        \item For all $c \in C$, $v \in V \setminus C$ and $j \in [k+1]$ we add an edge $\{c, v^j\}$ to our new edge set $E'$ if $\{c,v\} \in E$.
        \item The distances $d': V' \rightarrow \mathbb{R}_{\geq 0}$ of the vertices in $V'$ are the same as for the respective vertices in $V$. To be more precise: $d'(v^i,w^l) = d(v,w)$, for all $v,w \in V$, $i, j \in [k+1]$, where we say that $c^j = c$ for any $c \in C$, $j \in [k+1]$.
    \end{itemize}
    An example of this transformation is depicted in Figure~\ref{fig:remove_centers}.

    It is easy to verify that the size of the graph $G'$ is bounded by $k+1$ times the size of $G$ and can be calculated in polynomial time. Now we use the $\alpha$-approximation for connected $k$-center to cluster the vertices in $V'$ into $k$ connected clusters $P_1',\ldots,P_k'$ with centers $c_1',\ldots,c_k'$ (note that these centers do not need to be contained in $C$). Let $r^*$ be the radius of the optimal assignment of the vertices in the original graph $G$ to $C$. Then the radius of our clustering is at most $\alpha \cdot r^*$ as this assignment can be transformed into a clustering for $G'$ with the same radius by taking $C$ as center set and assign all vertices in $\{v^j\mid v \in V \setminus C\land j \in [k+1]\}$ according to the assignment of the vertices in $V \setminus C$. We will use the clustering produced by the approximation algorithm to create an assignment of $V$ to $C$ while increasing its radius only by a factor of $2$.

    For any $j \in [k+1]$ let $V^j = \{v^j\mid v \in V \setminus C\}$. Then we know that there exists at least one $j^* \in [k+1]$ such that $V^{j^*}$ contains no center $\{c_1',\ldots,c_k'\}$. For every $i \in [k]$ with $P_i' \cap (V^{j^*}\cup C) \neq \emptyset$ it holds that every vertex in $P_i' \cap (V^{j^*}\cup C)$ has distance at most $\alpha \cdot r$ to $c_i'$ which together with the triangle inequality also implies that every vertex has only distance $2 \alpha \cdot r^*$ to any center in $C \cap P_i'$. Let $S_1^{j^*},\ldots,S_l^{j^*}$ be the connected components of $P_i' \cap (V^{j^*} \cup C)$. Then each of these components contains at least one vertex in $C$, as each component is connected to the center $c_i' \not\in V^{j^*}$ and $C$ separates $V^{j^*}$ from the rest of the graph. We can assign for any $i' \in [l]$ the vertices in $S_{i'}^{j^*}$ to the respective center in $C \cap P_i'$ which is closest to it regarding the shortest path distance within the component $S_{i'}^{j^*}$. Then every vertex on the respective shortest path gets also assigned to the same center and we end up with disjoint connected components with the centers $C \cap P_i'$ that contain all vertices in $P_i' \cap V^{j^*}$. 
    
    %As a result we can assign all vertices in $C_i' \cap (V^{j^*} \cup C)$ to the centers in $C \cap C_i'$ by assigning the vertices in $S_{i'}^{j^*} \setminus C$ to one of the centers in $S_{i'}^{j^*} \cap C$ and make the other centers in $S_{i'}^{j^*}$ singletons for any $i' \in [l]$. One may verify that the resulting clusters are connected.

    By doing this for any $P_i'$ with $i \in [k]$ we assign every vertex in $V^{j^*}$ to exactly one center in $C$, since $P_1',\ldots,P_k'$ cover $V^{j^*}$ and all of these clusters are pairwise disjoint. As argued above, the resulting clusters $P_1,\ldots,P_k$ of $c_1,\ldots,c_k$ are connected and the radius is limited by $2 \alpha \cdot r^*$. Given that the subgraph of $G'$ on $(V^{j^*} \cup C)$ is isomorphic to $G$, this corresponds to a $2 \alpha$-approximation of the optimal assignment to $C$.
\end{proof}

In total we obtain that it is NP-hard to approximate both the assignment version and the regular version of connected $k$-center with an approximation ratio in $o(log^*(k))$:

% \begin{theorem}
%     It is NP-hard to approximate connected $k$-center with an approximation ratio in $o(\log^*(k))$. The same is also true for the assignment version.
% \end{theorem}

\ThmHardnessCenter*

\begin{proof}
    By Theorem~\ref{thm:reduction_centers_assignment} it is sufficient to prove the claim for the assignment version as any $o(\log^*(k))$-approximation algorithm for the non-assignment version also corresponds to such an algorithm for the assignment version. Suppose there existed a polynomial time algorithm calculating an $f(k)$-approximation for the assignment version of connected $k$-center with $f(k) \in o(\log^*(k))$. Since Lemma~\ref{lem:reduction_bound_prod} tells us that $L \in \Omega(\log^*(\mathcal{P}_L))$, we may conclude that there exists a constant $L' \in \mathbb{N}$ such that $f(\mathcal{P}_{L'}) < 2L'$. 
    
    For any given 3-SAT formula $\phi$ we know by Lemma~\ref{lem:red_center_polynomial} that we can calculate the instance $I(0,([\mathcal{S}_1],\ldots,[\mathcal{S}_{L'}]))$ and execute the approximation algorithm on it in time polynomially bounded in the size of $\phi$. By Lemma~\ref{lem:red_center_upper} we know that if $\phi$ can be satisfied, there exists an assignment of radius $1$. Thus the algorithm returns a solution with radius at most $f(\mathcal{P}_{L'}) < 2L'$. At the same time if $\phi$ cannot be satisfied, Lemma~\ref{lem:red_center_lower} tells us that any assignment to the centers (including the one calculated by the algorithm) has at least radius $2L'$. Thus one could solve 3-SAT in polynomial time by checking whether the solution of the algorithm has a radius strictly smaller than $2L'$ or a radius of at least $2L'$. This would imply P = NP.
\end{proof}

\section{Connected min-sum-radii} \label{sec:msr}

To approximate the connected min-sum-radii problem, we slightly adapt the primal dual algorithm by Buchem et al.~\cite{buchem20243+} which calculates a $(3 + \epsilon)$-approximation for the regular min-sum-radii problem. For a fixed $\epsilon >0$, it first guesses the $ \frac{1}{\epsilon}$ largest clusters and then approximates the remaining smaller clusters via a primal dual approach. The guessing of the largest clusters can be done by just going over all possible centers and radii and assigning all vertices within the respective radius to the centers. However, if we introduce the connectivity constraint we cannot hope to obtain the exact optimal clusters this way as it gets a lot harder to decide where we need to assign the vertices. Instead we can only obtain enclosing balls around the clusters and later have to deal with overlap between different clusters by merging them. Buchem et al.\ faced similar problems when considering lower bound constraints for which they provided a $(3.5 + \epsilon)$-approximation. However their solution does not work for the connectivity constraint, as it requires that any superset of a feasible cluster is again feasible. Using a slightly different approach that would in turn not work for non-uniform lower bound, we are able to obtain a $(3 + \epsilon)$-approximation for the connected min-sum radii problem.

For a vertex $v$ and a radius $r$ we define $B(v,r)$ as the set of all vertices in $V$ that are reachable from $v$ via a path only using vertices whose distance to $v$ is bounded by $r$. We denote the set of all pairs $(v,r)$ of vertices $v \in V$ and distances $r = d(v,w)$ for another vertex $w \in V$ as $\mathcal{S}$. Obviously $|\mathcal{S}| \leq |V|^2$. For an optimal connected cluster $P_i$ with radius $r^*$ and center $c^*$ it holds that $P_i \subseteq B(c^*,r^*)$. For an fixed $\epsilon > 0$, we can guess the centers $c_1^*,...,c_{1/\epsilon}^*$ and radii $r_1^*,...,r_{1/\epsilon}^*$ of the $1/\epsilon$ largest clusters (in decreasing order) and define $V' = V \setminus \bigcup_{i=1}^{1/\epsilon} B(c_{i}^*,r_{i}^*)$. We define $\B = \{ (v,r) \in \mathcal{S}\mid r \leq r_{1/\epsilon}^*\}$. One may note that for every remaining optimal cluster with center $c^*$ and radius $r^*$ it holds that $r^* \leq r_{1/\epsilon}^*$ and thus $(c^*,r^*) \in \B$. Let $\opt$ denote the cost of the optimal connected min-sum-radii solution and let $\opt' = \opt - \sum_{i=1}^{1/\epsilon} r_i^*$ be the cost produced by the $k' = k -{1/\epsilon}$ smallest clusters. We might observe that it is possible to cover the vertices in $V'$ with cost $\opt'$ by choosing the $k'$ pairs in $\B$ corresponding to the remaining optimal clusters:

\begin{observation}\label{obs:split_opt}
    There exists $k'$ pairs $(c_1,r_1),...,(c_{k'},r_{k'}) \in \B$ such that:
    \begin{itemize}
        \item $V' \subseteq \bigcup_{i=1}^{k'} B(c_i,r_i)$
        \item $\sum_{i=1}^{k'} r_i \leq \opt'$
    \end{itemize}
\end{observation}

Together with the ${1/\epsilon}$ largest clusters these pairs basically correspond to a connected min-sum-radii solution of $G$ where the clusters could overlap. Unfortunately we cannot calculate these pairs efficiently. However we can relax this problem by reformulating it as an LP and dropping the integrality constraint. We obtain:
\begin{equation}\label{LP_primal}
    \begin{array}{rrll}
        \min& \multicolumn{2}{c}{ \sum\limits_{(v,r) \in \B} r \cdot x_{v,r} } \\
        \textrm{s.t.}&\sum\limits_{(v,r) \in \B: u \in B(v,r)} x_{v,r} & \geq 1, & \forall u \in V',\\ 
        &\sum\limits_{(v,r) \in \B} x_{v,r} &\leq k' \\
        &x_{v,r} &\geq 0, & \forall (v,r) \in \B, \\
    \end{array}
\end{equation}
And the corresponding dual linear program:
\begin{equation}\label{LP_dual}
    \begin{array}{rrll}
        \max& \multicolumn{2}{c}{ \sum\limits_{u \in V'} \alpha_u - k' \cdot \lambda } \\
        \textrm{s.t.}&\sum\limits_{u \in B(v,r) \cap V'} \alpha_u & \leq r + \lambda, & \forall (v,r) \in \B,\\ 
        &\alpha_u &\geq 0, & \forall u \in V', \\
        & \lambda &\geq 0
    \end{array}
\end{equation}
Using our observation it is clear that the optimal value of these LPs is upper bounded by $\opt'$. 

In the following we need to present some definitions from the work of Buchem et al.\ which sometimes are a little bit modified to better work with our analysis.

We say that for a solution $(\lambda,\alpha)$ of LP~\eqref{LP_dual} the dual constraint corresponding to a pair $(v,r) \in \B$ is almost tight if $\sum\limits_{u \in B(v,r) \cap V'} \alpha_u  \geq r + \lambda - \mu$, for $\mu = \frac{r_{1/\epsilon}^*}{|V|^2}$. To simplify notation, we say then that the pair $(v,r)$ itself is \emph{almost tight}. 

Furthermore for a set of pairs we are interested in the connected components of this set where we basically consider two pairs neighbored if their corresponding balls share a common vertex.

\begin{definition}
    Let $\mathcal{S}' \subseteq \mathcal{S}$ denote a set of pairs. We define an overlap graph $G_o = (V_o,E_o)$ with $V_o = \{v_{u,r}\mid (u,r) \in \mathcal{S}'\}$ and $E_o = \{\{v_{u,r},v_{u',r'}\}\mid B(u,r) \cap B(u',r') \neq \emptyset \}$. A set of pairs $(u_1,r_1),\ldots,(u_\ell,r_\ell) \in \mathcal{S}'$ forms a connected component of $\mathcal{S'}$ if the corresponding vertices form a connected component of $G_o$, and $comp(\mathcal{S}')$ denotes the collection of components of $\mathcal{S}'$.
\end{definition}

We are now interested to find an LP solution for the dual LP together with a set of pairs that fulfills some structural properties:

\begin{definition}\label{def:well_structured}
    Let $(\lambda,\alpha)$ be a solution to~\eqref{LP_dual} and let $\B' \subseteq \B$. We say that $\B'$ is a set of \emph{structured pairs} for $(\lambda,\alpha)$ if there is an almost tight pair $(u',r') \in (\B\setminus \B')$ such that,
    \begin{enumerate}
        \item each $(u,r) \in \B'$ is almost tight in $(\lambda,\alpha)$,
        \item $\B' $ covers $V'$,
        \item $|comp(\B' )| \geq k' \geq |comp(\B' + (u',r'))|$
    \end{enumerate}
\end{definition}

One may note that in the work by Buchem et al.\ $(u',r')$ itself was also a part of $\B'$. We changed this to simplify notation.

As Buchem et al.\ showed, one can obtain a solution for the dual LP as well as a set of structured pairs in polynomial time by repeatedly raising some dual variables until an additional constraint becomes almost tight. The procedure still works if one uses connected balls instead of just all points within the radius of the center.

\begin{lemma}[Lemma 2.2 in \cite{buchem20243+}]
    In polynomial time, we can compute a solution $(\lambda,\alpha)$ of~\eqref{LP_dual}, together with a set $\B'$ of structured pairs for $(\lambda,\alpha)$.
\end{lemma}

The main idea is now to merge the balls corresponding to the pairs of the same component together. Given that the neighbored pairs share a common vertex we are able to bound the radius of the resulting cluster with the triangle inequality by the radii of the contained pairs. Additionally since the pairs are almost tight, we know that their radius can be lower bounded by the sum of the $\alpha_u$ values of the vertices $u$ contained in the cluster, which are part of the objective function. However, neighbored pairs share common vertices which means that these vertices can contribute in making both of these pairs almost tight. To circumvent this, Buchem et al.\ only considered sets of pairs whose balls are pairwise disjoint. For a component $\Com$ of pairs of a subset $\mathcal{S}' \subseteq \mathcal{S}$ we define $rad(\Com)$ to be the radius of the cluster $\bigcup_{(u,r) \in \Com} B(u,r)$ if we choose the best possible vertex in it as a center. For any possible subset $S \subseteq \mathcal{S}$ (which might correspond to a component) we define $sr(S)$ to be the sum of the radii of the contained pairs, i.e.\ $sr(S) = \sum_{(u,r) \in S} r$, and $dsr(S) = sr(S_d)$ where $S_d$ is a subset of pairs in $S$ such that the respective balls of the pairs are pairwise disjoint maximizing $sr(S_d)$. As Buchem et al.\ showed, we can bound the radius $rad(\Com)$ of a component of pairs by the sum of the radii of these disjoint sets:

\begin{lemma}[First half of Lemma 2.4 in \cite{buchem20243+}]\label{lem:bound_radius}
    For any component $\Com$ of pairs in $\mathcal{S'}$ it holds that $rad(\Com) \leq 3 \cdot dsr(\Com)$.
\end{lemma}

%Furthermore since we are talking about the sum of the radius of pairs whose clusters do not overlap when considering the value $dsr(\Com)$ of a component $\Com$ of pairs in $\B'$ using the LP solution $(\alpha,\lambda)$. 

Furthermore since we are talking about the sum of the radii of pairs whose clusters do not overlap when considering $dsr$, we can also bound the respective values by the LP solution $(\alpha,\lambda)$ (for pairs in $\B'$):

\begin{lemma}[Second half of Lemma 2.4 in \cite{buchem20243+}]\label{lem:bound_dsr}
    Let $\Com$ be a connected component of $\mathcal{B}'$. Then:
    \begin{equation*}
        dsr(\Com) \leq \left(\sum_{v \in V'(\Com)} \alpha_v\right) - \lambda + |V'(\Com)| \mu,
    \end{equation*},
    where $ V'(\Com)$ denotes the set $V' \cap \left(\bigcup_{(c,r) \in \Com} B(c,r)\right)$.
\end{lemma}

A very important property of the set $\B'$ of structured pairs is that it has at least $k'$ components. This is relevant as the objective function of LP~\eqref{LP_dual} subtracts $\lambda$ exactly $k'$ times which also needs to be done if we want to bound the cost of our solution by the LP cost. Intuitively, this ensures that the LP solution actually uses $k'$ clusters and does not increase the cost of the vertices by forcing them into bigger clusters.

However, to obtain at most $k'$ clusters we need to merge all components overlapping with the ball corresponding to the pair $(u',r')$ in Definition \ref{def:well_structured}. Additionally there might also be some overlap with the pairs $(c_1^*,r_1^*),\ldots,(c_{1/\epsilon}^*,r_{1/\epsilon}^*)$ that we have to deal with. Let
\begin{equation*}
    \mathcal{S}' = \B' \cup \{(u',r')\} \cup \{(c_i^*,r_i^*)\mid i \in [1/\epsilon]\}.
\end{equation*}
Then our algorithm merges together the balls corresponding to pairs of the same component of $comp(\mathcal{S}')$ and choose the vertices minimizing the respective radius of the clusters as centers. One may note that if two overlapping clusters are connected and share a common vertex that then the merged cluster is also connected. Thus the algorithm actually calculates a connected min-sum-radii solution with at most $k$ clusters.

\begin{observation}
    Let $P_1, P_2 \subseteq V$ be connected clusters of $G$ with $P_1 \cap P_2 \neq \emptyset$. Then $P_1 \cap P_2$ is also connected.
\end{observation}

A critical property to bound the cost of the resulting clusters is that if some pairs in $\mathcal{S}' \setminus \B'$ connect components of $\B'$ to a single component $\Com$ then $dsr(\Com)$ is bounded by the $dsr$ values of the previous components plus the radii of the additional pairs:

\begin{lemma}\label{lem:combine_components}
Let $\Com^1,...,\Com^\ell$ be connected components of $\mathcal{B}' $. Let $D\subseteq \mathcal{S}$ be a set of pairs connecting them to a single component $\Com'$. Then
\begin{equation*}
    dsr(\Com') \leq \sum_{i = 1}^\ell dsr(\Com^i) + sr(D).
\end{equation*}
\end{lemma}

\begin{proof}
    Let $\Com_d'$ be a subset of pairs in $\Com'$ whose corresponding balls are disjoint that maximizes $sr(\Com_d)$. Then:
    \begin{align*}
        dsr(\Com') &= sr(\Com_d)\\
        &= sr(\Com_d' \cap D) + \sum_{i= 1}^{\ell}sr(\Com_d'\cap \Com^i)\\
        &\leq sr(D) + \sum_{i= 1}^{\ell} dsr(\Com^i).
    \end{align*}
\end{proof}

Using this, we can bound the radii of the resulting clusters which can then be used to bound the cost of our algorithm:

\begin{lemma}\label{lem:msr_bound_cluster}
    Let $\Com'$ be a connected component of $\mathcal{S}'$, $D = \Com' \setminus \mathcal{B}' $ and let $\Com^1,\ldots,\Com^\ell$ be the components of $\Com' \cap \mathcal{B}' $. Then
    \begin{equation*}
        dsr(\Com') \leq \left(\sum_{v \in V'(\Com')} \alpha_v\right) - \ell\cdot \lambda + |V'(\Com')| \mu + sr(D).
    \end{equation*}
\end{lemma}

\begin{proof}
    By combining Lemma~\ref{lem:combine_components} and~\ref{lem:bound_dsr}, we obtain
    \begin{align*}
        dsr(\Com') & \stackrel{\text{Lem. }\ref{lem:combine_components}}{\leq} \sum_{i = 1}^\ell dsr(\Com^i) + sr(D)\\
        & \stackrel{\text{Lem. }\ref{lem:bound_dsr}}{\leq} \sum_{i = 1}^\ell \left(\left(\sum_{v \in V'(\Com^i)} \alpha_v\right) - \lambda + |V'(\Com^i)| \mu \right) + sr(D)\\
        & \stackrel{\phantom{\text{Lem. }\ref{lem:bound_dsr}}}{\leq} \left(\sum_{v \in V'(\Com')} \alpha_v\right) - \ell\cdot \lambda + |V'(\Com')| \mu + sr(D),
    \end{align*}
    where in the last inequality we used that the vertex sets corresponding to $\Com^1,\ldots,\Com^\ell$ are pairwise disjoint. 
\end{proof}

% \begin{theorem}\label{thm:msr_connected}
%     The algorithm calculates a $3 + O(\epsilon)$ approximation for connected min-sum-radii.
% \end{theorem}

\ThmMSR*

\begin{proof}
    Let $\A$ be the connected components of $\mathcal{S}'$ and for any $A \in \A$ let $\ell_A$ be the number of connected components of $\B'$ contained in $A$. Then we can bound the cost of the solution of our algorithm as follows:
    \begin{align}
        \sum_{A \in \A} rad(A) &\stackrel{\text{Lem. }\ref{lem:bound_radius}}{\leq} 3 \sum_{A \in \A} dsr(A)\\
        &\stackrel{\text{Lem. }\ref{lem:msr_bound_cluster}}{\leq} 3 \sum_{A \in \A} \left( \left(\sum_{v \in V'(A)} \alpha_v\right) - \ell_A \cdot \lambda + |V'(A)| \mu + sr(A \setminus \B')\right). \label{eq:bound_cost_msr}
        %&\stackrel{\phantom{\text{Lem. }\ref{lem:msr_bound_cluster}}}{=} 3\left( \sum_{v \in V'} \alpha_v - |comp(\B')|\lambda + |V'| \mu + sr(\mathcal{S}' \setminus \B')\right)\\
        %&\stackrel{\phantom{\text{Lem. }\ref{lem:msr_bound_cluster}}}{\leq} 3\left( \sum_{v \in V'} \alpha_v - k'\lambda + 2 \epsilon \opt + \sum_{i= 1}^{1/\epsilon}r_i^*\right)\\
        %&\stackrel{\phantom{\text{Lem. }\ref{lem:msr_bound_cluster}}}{\leq} 3\left( \opt' +  \sum_{i= 1}^{1/\epsilon}r_i^*\right) + 6 \epsilon \opt\\
        %&\stackrel{\text{Obs. }\ref{obs:split_opt}}{\leq} (3 + 6 \epsilon) \opt
    \end{align}

    We would now like to get rid of the sum over all connected components $A \in \A$. To do this, we will look at the different terms in the large brackets of Inequality \eqref{eq:bound_cost_msr} and bound their total value in this sum.

    For every $v \in V'$ we know that it is contained in exactly one component $A \in \A$. As a result $\sum_{A \in \A}  \sum_{v \in V'(A)} \alpha_v = \sum_{v \in V'} \alpha_v$  and $\sum_{A \in \A} |V'(A)|\mu = |V'| \mu $.
    
    Since $\B' \subseteq \mathcal{S}'$ we have that every connected component of $\B'$ is contained in exactly one of the components $A \in \A$ and thus 
    \begin{equation*}
        \sum_{A \in \A} \ell_A \lambda= comp(\B')\lambda \geq k' \lambda,
    \end{equation*}
    where we used in the last inequality that $\B'$ is a set of structed pairs.
    
    At the same time $\bigcup_{A \in \A} A = \mathcal{S}'$ and every pair in $\mathcal{S}' \setminus \B'$ is contained in exactly one component $A \in \A$. Thus 
    \begin{equation*}
    \sum_{A \in \A} sr(A \setminus \B') = sr(\mathcal{S}' \setminus \B') = r' + \sum_{i= 1}^{1/\epsilon}r_i^*
    \end{equation*}
    
    By inserting these bounds into Inequality~\eqref{eq:bound_cost_msr} we obtain:
    
    \begin{equation*}
        \sum_{A \in \A} rad(A) \leq 3\left( \left(\sum_{v \in V'} \alpha_v\right) - k'\lambda + |V'| \mu + r'+ \sum_{i= 1}^{1/\epsilon}r_i^*\right).
    \end{equation*}

    One might note that $\left(\sum_{v \in V'} \alpha_v\right) - k'\lambda$ is exactly the objective value of the dual LP and by Observation~\ref{obs:split_opt} thus bounded by $\opt'$. Additionally by the choice of $\mu$ we know that $|V'|\mu \leq \epsilon \cdot \opt$. Lastly we know that $r' \leq r_{1/\epsilon^*}$ or in other words $r'$ is bounded by the radius of the $1/\epsilon$ largest cluster which by Markov's inequality is upper bounded by $\epsilon \cdot \opt$. As a result:
    \begin{equation*}
        \sum_{A \in \A} rad(A) \leq 3\left( \opt' +  \sum_{i= 1}^{1/\epsilon}r_i^*\right) + 6 \epsilon \cdot \opt \leq (3 + 6 \epsilon) \opt,
    \end{equation*}
    which proves the theorem
    \end{proof}

\subsection{Connected min-sum-diameter} \label{sec:msd}

In the following, we consider the problem of finding a connected min-sum-diameter solution. Buchem et al.\ pointed out that their algorithm trivially also calculates a $(6 + \epsilon)$-solution for the unconnected min-sum-diameter variant, as there is at most a gap of two between the min-sum-radii and min-sum-diameter objective~\cite{buchem20243+}. While they did not put much focus on the min-sum-diameter objective, other authors also provided results for this objective that were not direct consequences of the respective mins-sum-radii algorithms. Among others  Behsaz and Salavatipour provided a PTAS for Eucledian MSD as well as an exact polynomial time algorithm for constant $k$ \cite{behsaz2015minimum} and Friggstad and Jamshidian proved a $6.546$-approximation for MSD in the same work in which they presented a $3.389$-approximation for MSR \cite{friggstad2022improved}. Interestingly enough the algorithm by Buchem et al.\ actually works quite well with the MSD objective as it can easily be shown that the resulting clusters directly form a $(4 + \epsilon)$-approximation in this setting, which is a significant improvement compared to the trivial $(6 + \epsilon)$-guarantee. To the author's knowledge, this has not been observed in the literature yet, which is why we will provide the respective proof here. It works both in the connected as well as the regular setting.

A very important observation is that for any point set with diameter $\delta$ we could choose an arbitrary vertex in it as a center to obtain a radius of at most $\delta$. As a result for any graph $G$ and metric $d$ the cost of the optimal connected MSD solution is upper bounded by the cost of the optimal MSR solution.

\begin{observation}
    For any graph $G$ and any metric $d$ and natural number $k \in \mathbb{N}$ the cost of the optimal connected min-sum-diameter solution is lower bounded by the cost of the optimal connected min-sum-radii solution.
\end{observation}

We will now show that the diameter of the clusters produced by the min-sum-radii algorithm is upper bounded by $4 + \epsilon$ times the cost of the optimal MSR solution. The main property that we need for this is that for any component $\Com$ of pairs the diameter $diam(\Com)$ of the cluster that gets created by merging the balls corresponding to the pairs in the component is upper bounded by $4 dsr(\Com)$ (which was also shown implicitly in \cite{buchem20243+}):

\begin{lemma}
    For any component $\Com$ of pairs in $\mathcal{S}'$ it holds that:
    \begin{equation*}
        diam(\Com) \leq 4 \cdot dsr(\Com)
    \end{equation*}
\end{lemma}

\begin{proof}
    Let $a$ and $b$ be the vertices in $V'(\Com)$ maximizing $d(a,b)$, i.e., $d(a,b) = diam(\Com)$. Let $(v_1,r_1),\ldots,(v_\ell,r_\ell)$ be a sequence of pairs such that $a \in B(v_1,r_1)$, $b \in B(v_\ell,r_\ell)$ and for any $i \in [\ell-1]$ the pairs $(v_1,r_1)$ and $(v_\ell,r_\ell)$ are neighbored, which means that the respective balls share a common vertex. We can ensure that if for two indices $1 \leq i< j \leq \ell$ the intersection $B(v_i,r_i) \cap B(v_j,r_j)$ is nonempty that $i = j - 1$ by removing the pairs $(v_{i+1},r_{i+1}),\ldots,(v_{j-1},r_{j-1})$ otherwise. In the paper by Buchem et al.\ such a sequence is called a nice path. We split this sequence up into two sets of pairs $H_1 = \{(v_i,r_i)\mid i \leq \ell \land $i$ \text{ is odd}\}$ and $H_2 = \{(v_i,r_i)\mid i \leq \ell \land $i$ \text{ is even}\}$. Then for each of these sets the balls corresponding to the contained pairs are pairwise disjoint and $sr(H_1) + sr(H_2) = \sum_{i= 1}^\ell r_i$. Thus:
    \begin{equation*}
        dsr(\Com) \geq \frac{1}{2} \sum_{i= 1}^\ell r_i.
    \end{equation*}
    At the same time we can use the triangle inequality to bound the distance from $a$ to $b$ as follows:
    \begin{align*}
        d(a,b) &\leq d(a,v_1) + \left(\sum_{i= 1}^{\ell - 1} d(v_i,v_{i+1})\right) + d(v_\ell,b)\\
        &\leq r_1 + \left(\sum_{i= 1}^{\ell - 1} r_i + r_{i+1}\right) + r_\ell\\
        &= 2 \sum_{i=1}^{\ell} r_i\\
        &\leq 4 dsr(\Com).
    \end{align*}
\end{proof}

Let now $\Com_1,\ldots,\Com_\ell$ be the connected components of the set $\mathcal{S}'$ of pairs produced by the MSR algorithm. As in the proof of Theorem~\ref{thm:msr_connected} we can upper bound the sum $\sum_{i=1}^\ell dsr(\Com_i)$ by $(1 + 2\epsilon)$ times the cost of the optimal MSR solution which is in turn upper bounded by the cost of the optimal MSD solution. Since for every component $\Com_i$ the diameter is then bounded by $4 \cdot dsr(\Com)$ we know that the sum of the diameters of all components is upper bounded by $(4 + 8 \epsilon)$ times the cost of the optimal MSD solution.

% \begin{theorem}
%     For any $\epsilon > 0$ one can calculate a $(4 + O(\epsilon))$-approximation for both the connected as well as the regular min-sum-diameter problem in polynomial time.
% \end{theorem}

\ThmMSD*

\bibliography{literature}

\end{document}